\documentclass[12pt]{article}
 \usepackage[margin=1in]{geometry} 

\usepackage{amsmath,amsthm,amssymb,amsfonts}
\usepackage{algorithm}
\usepackage{algorithmic}
\usepackage{todonotes}

\newcommand\bigforall{\mbox{\Large $\mathsurround0pt\forall$}}

\newtheorem{theorem}{Theorem}
\newtheorem{definition}{Definition}
\newtheorem{lemma}{Lemma}
\newtheorem{claim}{Claim}

\newtheorem{corollary}{Corollary}

\begin{document}
 
\title{Dual-Based Approximation Algorithms for Cut-Based Network Connectivity Problems}
\author{Benjamin Grimmer\\ {\small bdg79@cornell.edu}}
\date{}
\maketitle


\begin{abstract}
We consider a variety of NP-Complete network connectivity problems. We introduce a novel dual-based approach to approximating network design problems with cut-based linear programming relaxations. This approach gives a $3/2$-approximation to {\sc Minimum 2-Edge-Connected Spanning Subgraph} that is equivalent to a previously proposed algorithm.
One well-studied branch of network design models ad hoc networks where each node can either operate at high or low power. If we allow unidirectional links, we can formalize this into the problem {\sc Dual Power Assignment} ({\sc DPA}). Our dual-based approach gives a $3/2$-approximation to {\sc DPA}, improving the previous best approximation known of $11/7\approx 1.57$. 

Another standard network design problem is {\sc Minimum Strongly Connected Spanning Subgraph} ({\sc MSCS}). We propose a new problem generalizing {\sc MSCS} and {\sc DPA} called {\sc Star Strong Connectivity} ({\sc SSC}). Then we show that our dual-based approach achieves a 1.6-approximation ratio on {\sc SSC}. As a consequence of our dual-based approximations, we prove new upper bounds on the integrality gaps of these problems.
\end{abstract}

\section{Introduction}
\label{sec:Intro}

In this work, we present approximation algorithms for multiple network connectivity problems. All the problems we consider seek a minimum cost graph meeting certain connectivity requirements. Problems of this type have a wide array of applications. They have uses in the design and modeling of communication and ad hoc networks. Often these problems involve balancing fault tolerance and connectivity against cost of building and operating a network.

Many standard network connectivity problems have been shown to be NP-Complete~\cite{Garey1979}. As a result, there is little hope of producing fast (polynomial time) algorithms to solve these problems (unless $P=NP$). So the focus has shifted to giving fast (polynomial time) algorithms that approximately solve these problems. An approximation algorithm is said to have an approximation ratio of $\alpha$ if the cost of its output is always within a factor of $\alpha$ of the cost of the optimal solution.
Technically, $\alpha$ can depend on the size of the problem instance, but a constant approximation ratio is better than one that grows.

In Subsections~\ref{sec:Intro-ECS} and~\ref{Intro-PowerAssignment}, we introduce a number of standard graph connectivity problems and their notable approximation algorithms. There are many other connectivity problems and  algorithmic tools that we omit from our discussion (See~\cite{Gupta2011} for a more in-depth survey).
Then in Subsection~\ref{sec:OurResults}, we summarize our contributions. In Section~\ref{sec:DualBased}, we introduce our novel approach to approximating network design problems. The remaining sections will present our approximation algorithms based on this methodology.

\subsection{Edge-Connectivity Problems}
\label{sec:Intro-ECS}

One standard graph connectivity problems is {\sc Minimum 2-Edge-Connected Spanning Subgraph} ({\sc 2ECS}). This problem takes as input a 2-edge-connected graph and outputs a 2-edge-connected spanning subgraph with minimum cardinality edge set.
A $5/4 = 1.25$-approximation was proposed using a matching lower bound in~\cite{Jothi2003}. A $4/3\approx 1.33$-approximation was given by Vempala and Vetta~\cite{Vempala2000}. Another notable approximation algorithm appears in~\cite{Khuller1994}, which achieves a $1.5$ ratio using a graph carving method with linear runtime. This graph carving method is a special case of the dual-based approach we introduce in Section~\ref{sec:DualBased}.

One can generalize the problem of {\sc 2ECS} to have weights on every edge. Then the output is the spanning subgraph with minimum total weight on its edges. This problem seems to be harder to approximate that the unweighted version. The best approximation known for the weighted problem was given by Khuller and Vishkin and achieves a 2-approximation in~\cite{Khuller1994}. Later, Jain proposed a 2-approximation that applies to a more general class of Steiner problems~\cite{Jain2001}.

Another generalization of {\sc 2ECS} is to search for the $k$-edge-connected spanning subgraph with the least number of edges. Using a linear program rounding algorithm, this problem on a multigraph input has a $1+3/k$-approximation if $k$ is odd and a $1+2/k$-approximation if $k$ is even~\cite{Gabow2009}. Further, in~\cite{Gabow2009}, it is shown that for any fixed integer $k\geq 2$ that a $(1+\epsilon)$-approximation cannot exist for arbitrary $\epsilon>0$ (unless $P=NP$).


One of the most fundamental directed graph connectivity problems is {\sc Minimum Strongly Connected Spanning Subgraph} ({\sc MSCS}). This problem takes as input a strongly connected digraph and outputs a strongly connected spanning subgraph with minimum cardinality arc set. In~\cite{Vetta2001}, Vetta proposes the best approximation known for {\sc MSCS} using a matching lower bound to get an approximation ratio of $1.5$. There are two other notable approximation algorithms for {\sc MSCS}. First, Khuller, Raghavachari and Young gave a greedy algorithm with a $1.61+\epsilon$ approximation ratio~\cite{Khuller1995}, ~\cite{Khuller1996}. Second, Zhao, Nagamochi and Ibaraki give an algorithm that runs in linear time with a $5/3\approx 1.66$-approximation ratio in~\cite{Zhao2003}. This algorithm implicitly uses the dual of the corresponding cut-based linear program to bound the optimal solution.

When {\sc MSCS} is generalized to have weights on each arc, the best algorithm known is a 2-approximation. This straightforward algorithm works by computing an in-arborescence and an out-arborescence with the same root in the digraph, and outputting their union.
This is the best algorithm known even when arc weights are restricted to be in $\{0,1\}$.

\subsection{Power Assignment Problems}
\label{Intro-PowerAssignment}

A well-studied area of network design focuses on the problem of assigning power levels to vertices of a graph to achieve some connectivity property. This is useful in modeling radio networks and ad hoc wireless networks. It is common in this type of problem to minimize total power consumed by the system. 
This class of problems take as input a directed simple graph $G=(V,E)$ and a cost function $c:E \rightarrow \mathbb{R}^+$. A solution to this problem assigns every vertex a non-negative power, $p(v)$. We use $H(p)$ to denote the spanning subgraph of $G$ induced by the power assignment $p$ (we will formally define $H(p)$ later). The minimization problem then is to find the minimum power assignment, $\sum p(v)$, subject to $H(p)$ satisfying a specific property.

The first work on Power Assignment was done by Chen and Huang~\cite{Chen1989}, which assumed that $E$ is bidirected. We say an instance of Power Assignment is bidirected if whenever $uv\in E$ then $vu\in E$ and $c(uv)=c(vu)$. There has been a large amount of interest in this type of problem since 2000 (some of the earlier papers are~\cite{Hajiaghayi2003},~\cite{Ramanathan2000} and~\cite{Wattenhofer2001}).
While we consider problems seeking strong connectivity, other works have focused on designing fault-tolerant networks. In~\cite{Wang2008}, approximations for both problems seeking biconnectivity and edge-biconnectivity are given. Further,~\cite{Lam2015} considers the more general problems of $k$-connectivity and $k$-edge-connectivity.

We consider an asymmetric version of Power Assignment that allows unidirectional links defined as follows. The power assignment induces a spanning directed subgraph $H(p)$ where $xy \in E(H(p))$ if the arc $xy \in E$ and $p(x) \geq c(xy)$. The goal is to minimize the total power subject to $H(p)$ being strongly connected.
This problem was shown to be NP-Complete in~\cite{Carmi2007}.
Many different approximations for this problem have been proposed, which are compared in~\cite{Calinescu2012}. If we assume the input digraph and cost function are bidirected, the best approximation known achieves a $1.85$-approximation ratio~\cite{Calinescu2013}.

We are particularly interested in a special case of asymmetric Power Assignment called {\sc Dual Power Assignment} ({\sc DPA}). This problem takes a bidirected instance of Asymmetric Power Assignment with cost function $c:E\rightarrow\{0,1\}$. This models a network where each node can operate at either high or low power, and finds a minimum sized set of nodes to assign high power to produce a strongly connected network. The best approximation known for {\sc DPA} was proposed by~\cite{Abu2014} and achieves an $11/7\approx1.57$-approximation. This algorithm is based on interesting properties of Hamiltonian cycles.
A greedy approach to approximating {\sc DPA} was first given in~\cite{Chen2005} and achieved a $1.75$-approximation ratio.
Then in C{\u{a}}linescu in~\cite{Calinescu2014} showed that this greedy approach can be extended to match the $1.61+\epsilon$-approximation ratio of Khuller et al. for {\sc MSCS} in~\cite{Khuller1995},~\cite{Khuller1996}.
A greedy algorithm based on the same heuristic was later shown to give a $5/3\approx 1.66$-approximation with nearly linear runtime in~\cite{Grimmer2014}.

\subsection{Our Results}
\label{sec:OurResults}

{\sc MSCS} and {\sc DPA} both have approximation algorithms based on very similar ideas. This raises the question of how these two problems are related. To answer this question, we propose a new connectivity problem generalizing both of them called {\sc Star Strong Connectivity} ({\sc SSC}), defined as follows:
We call a set of arcs sharing a source endpoint a \emph{star}.
{\sc SSC} takes a strongly connected digraph $G=(V,E)$ and a set $\mathcal{C}$ of stars as input such that $\bigcup_{F\in \mathcal{C}} F=E$. Then {\sc SSC} finds a minimum cardinality set $R\subseteq \mathcal{C}$ such that $(V,\bigcup_{F\in R}F)$ is strongly connected.

{\sc SSC} is exactly {\sc MSCS} when all $F\in \mathcal{C}$ are restricted to have $|F|=1$. Under this restriction, choosing any star in {\sc SSC} would be equivalent to choosing its single arc in {\sc MSCS}. Further, we make the following claim relating {\sc SSC} and {\sc DPA} (proof of this lemma is deferred to Section~\ref{sec:1.5DPA}).

\begin{lemma} \label{lem:DPATOSSC}
When {\sc Star Strong Connectivity} has a bidirected input digraph $G$ (an arc $uv$ exists if and only if the arc $vu$ exists), it is equivalent to {\sc Dual Power Assignment}.
\end{lemma}

In some sense, {\sc SSC} has a more elegant formulation than {\sc DPA}. It removes the complexity of having two different classes of arcs. This benefit becomes very clear when constructing Integer Linear Programs for the two problems (our linear programs for {\sc SSC} and {\sc DPA} will be formally introduced in Section~\ref{sec:1.5DPA}). Besides this difference in elegance, both the resulting programs for {\sc DPA} and {\sc SSC} constrained to have a bidirected input digraph are equivalent to each other.

We introduce a novel dual-based approach to approximating connectivity problems. This methodology utilizes the cut-based linear programming relaxation of a connectivity problem.
In Section~\ref{sec:DualBased}, we give a full description of our approach and apply it to the problem of {\sc 2ECS}. The resulting algorithm is equivalent to the 3/2-approximation given in~\cite{Khuller1994}. Its value is solely in its simplicity and serving as an example of our approach.

Applying our dual-based method to {\sc DPA} gives a tight 3/2-approximation. This improves the previous best approximation known for {\sc DPA} of $11/7\approx 1.57$~\cite{Abu2014}. Our algorithm and its analysis are made simpler by viewing it as an instance of {\sc SSC} with a bidirected input digraph instead of {\sc DPA}. In Section~\ref{sec:1.5DPA}, we present our algorithm, prove its approximation ratio and show this bound is tight. 

\begin{theorem} \label{thm:1.5DPA}
{\sc Dual Power Assignment} has a dual-based 1.5-approximation algorithm.
\end{theorem}

Integer Linear Programs are often used to formulate NP-Complete problems. An approximation that uses a linear programming relaxation typically cannot give a better approximation ratio that the ratio between solutions of the integer and relaxed programs. This ratio is known as the integrality gap of a program. For minimization problems, it is formally defined as the supremum of the ratio between the optimal integer solution and the optimal fractional solution over all problem instances.
As a result of our analysis for Theorem~\ref{thm:1.5DPA}, we prove an upper bound of 1.5 for {\sc DPA}'s integrality gap. This improves the previous upper bound of 1.85 proven in~\cite{Calinescu2013}.

\begin{corollary} \label{cor:1.5DPA-IntegralityUpper}
The integrality gap of the standard cut-based linear program for {\sc Dual Power Assignment} is at most 1.5.
\end{corollary}

In fact, we prove a slightly stronger statement. Our algorithm constructs integer primal and integer dual solutions. As a result, the gap between integer solutions of these two problems is at most 1.5.

Now, we turn our focus to the more general problem of {\sc SSC}. Our algorithm for {\sc DPA} is dependent on the underlying digraph being bidirected. Further the 3/2-approximation for {\sc MSCS} in~\cite{Vetta2001} and 11/7-approximation for {\sc DPA} in~\cite{Abu2014} do not seem to generalize to {\sc SSC}. However the greedy approach used by Khuller et al. in~\cite{Khuller1994} and~\cite{Khuller1995} on {\sc MSCS} and by C{\u{a}}linescu in~\cite{Calinescu2014} on {\sc DPA} appears to generalize easily to {\sc SSC}.

\begin{claim}
{\sc SSC} has a $1.61+\epsilon$ polynomial approximation algorithm using a simple variation of the greedy algorithm of~\cite{Khuller1995}.
\end{claim}

In Section~\ref{sec:1.6SSC}, we improve on this $1.61+\epsilon$-approximation by applying our dual-based approach to {\sc SSC}. This produces an algorithm with a tight 1.6-approximation ratio. Since {\sc MSCS} is a subproblem of {\sc SSC}, this approximation ratio also extends to it. 
As with our approximation of {\sc DPA}, we observe that an upper bound on the integrality gap follows from our analysis.

\begin{theorem} \label{thm:1.6SSC}
{\sc Star Strong Connectivity} has a dual-based 1.6-approximation algorithm.
\end{theorem}
\begin{corollary} \label{cor:1.6SSC-IntegralityUpper}
The integrality gap of the standard cut-based linear program for {\sc Star Strong Connectivity} (and thus {\sc MSCS}) is at most 1.6.
\end{corollary}

Again, we actually prove a slightly stronger statement. Since our algorithm constructs integer primal and integer dual solutions, the primal-dual integer gap of these two problems is at most 1.6.

A lower bound on the integrality gap of a problem provides a bound on the quality of approximation that can be achieved with certain methods. Linear program rounding, primal-dual algorithms, and our dual-based algorithms are all limited by this value. Recently, Laekhanukit et al.\ proved the integrality gap of {\sc MSCS} is at least $3/2-\epsilon$ for any $\epsilon>0$~\cite{Laekhanukit2012}.


\section{Dual-Based Methodology}
\label{sec:DualBased}

In Subsection~\ref{sec:DualBased-Intro}, we describe the typical form of integer linear programs (ILPs) related to graph connectivity problems and give a high-level description of our dual-based approach for a general problem. Finally, we apply our dual-based approach to the problem of {\sc 2ECS} as a simple example.

\subsection{Cut-Based Linear Programs}
\label{sec:DualBased-Intro}


All the connectivity problems considered in this work have cut-based ILPs. In this type of program, the constraints are based on having at least a certain number of edges or arcs crossing each cut of the graph. We will give the standard cut-based ILP for {\sc Minimum 2-Edge-Connected Spanning Subgraph} to demonstrate this structure.
For a cut $\emptyset\subset S\subset V$, we use $\partial E(S)$ to denote all edges with exactly one endpoint in $S$. Then the standard cut-based linear programming relaxation for {\sc 2ECS} is the following.

\begin{alignat*}{2}
	\text{{\sc 2ECS Primal LP}} \\
	\text{minimize }   & \sum_{e\in E} x_e \\
    \text{subject to } & \sum_{e\in \partial E(S)}x_e\geq 2 \ & , \ & \bigforall \emptyset\subset S\subset V\\
                       & x_e\geq 0 \ & , \ & \bigforall e\in E
\end{alignat*}

The integer programming formulation for {\sc 2ECS} is given by adding the constraint that all $x_e$ are integer valued. The linear program will always have objective less than or equal to the objective of the integer program. 

It is worth noting that this type of program has an exponential number of constraints, but it could be converted into a polynomial sized program using flow-based constraints. Previous algorithms have taken advantage of polynomial time linear program solvers to approximate {\sc 2ECS} using its linear programming relaxation. However, our dual-based algorithms do not need to solve the linear program. As a result, we can keep it in the simpler cut-based form.

Our method takes advantage of the corresponding dual linear program. The dual program will always have objective at most that of the original (primal) program. This property is known as weak duality. In fact, their optimal solutions will have equal objective, but we do not need this stronger property. The dual program corresponding to {\sc 2ECS} is the following.

\begin{alignat*}{2}
	\text{{\sc 2ECS Dual LP}} \\
	\text{maximize }   & \sum_{\emptyset\subset S\subset V} 2y_S \\
    \text{subject to } & \sum_{e\in \partial E(S)} y_S\leq 1 \ & , \ & \bigforall e\in E\\
                       & y_S\geq 0 \ & , \ & \bigforall \emptyset\subset S\subset V
\end{alignat*}

Since the optimal solution to our integer program is lowerbounded by that of the primal linear program, we know that the optimal integer solution is lowerbounded by every feasible solution to the dual program.

The basic idea of our dual-based method is to build a feasible dual solution while constructing our integer primal solution. We construct our primal solution by repeatedly finding and contracting a problem specific type of subgraph: a cycle for {\sc 2ECS}, a perfect set for {\sc SSC} (defined later).
Using a cut-based linear program, a dual solution will be a set of disjoint cuts (where the exact definition of disjoint is problem specific).
We are interested in cuts that are disjoint from all cuts after contracting a subgraph. Later, we formally define these as internal cuts. We choose the subgraph to contract based on it having a number of disjoint internal cuts. When our algorithm terminates, the union of these disjoint internal cuts in each iteration will give a feasible dual solution.

To apply this to a new connectivity problem, we first must define a contractible subgraph such that repeated contraction will yield a feasible primal solution. Definitions for disjoint and internal cuts will follow from the cut-based program and the contractible subgraphs. Finally, any polynomial runtime procedure that constructs a contractible subgraph with at least one internal cut produces an algorithm creating integer primal and dual solutions. The quality of the approximation depends directly on the number of internal cuts in each contraction. Our dual solution could use fractional cuts. However, this did not result in improvements in the approximation bounds for the problems studied in this paper.

\subsection{Application to 2ECS}
\label{sec:1.5-2ECS}

We will illustrate our dual-based approach by giving a straightforward 3/2-approximation to {\sc 2ECS}. This is neither best known in approximation ratio nor runtime. Its value is to serve as a simple example of our dual-based approach. In~\cite{Khuller1994}, an equivalent 3/2-approximation is given for {\sc 2ECS} that implicitly uses the dual bound. They also show that simple modifications of the algorithm allow it to run in linear time.

Our approximation is based on iteratively selecting and contracting cycles in the graph until the graph is reduced to a single vertex (this approach has been used by multiple previous approximations).
We claim that such a procedure will always produce a 2-edge-connected spanning subgraph. Consider any cut $\emptyset \subset S \subset V$ in the graph. At some point, we will select a cycle with vertices in both $S$ and $V \setminus S$. This cycle must have two edges crossing the cut. Thus such a procedure will always produce a feasible solution.

Consider the primal and dual programs for {\sc 2ECS} given in Section~\ref{sec:DualBased-Intro}. When the dual problem is restricted to integer values, it can be thought of as choosing a set of cuts such that no two cuts share any edges. Our algorithm builds a solution to this dual problem to lower bound the optimal primal solution.

We say a cut $S$ is \emph{internal} to a cycle $C$ if all edges in $\partial E(S)$ have both endpoints in the vertices of $C$. To contract a cycle means to replace all vertices of the cycle with a single supervertex whose edge set is all edges with exactly one endpoint in the cycle. When contracting a cycle, we keep duplicate edges (and thus the resulting structure is a 2-edge-connected multigraph). Contracting a cycle with an internal cut will remove all edges crossing that cut from the graph. As a result, after repeated contraction of cycles each with an internal cut, the set of all these internal cuts is dual feasible.
Following from this, our algorithm will find a cycle with an internal cut, add the edges of the cycle to our approximate solution, contract the vertices of our cycle, and then repeat.
Complete description of this process is given in Algorithm~\ref{alg:1.5-2ECS}.

\begin{algorithm}
\caption{Dual-Based Approximation for {\sc 2ECS}} \label{alg:1.5-2ECS}
\begin{algorithmic}[1]
  \STATE $R=\emptyset$
  \WHILE{$|V|\neq 1$}
    \STATE Find a cycle $C$ with an internal cut as shown in Lemma~\ref{lem:2ECS-Construction}
	\STATE Contract the vertices of $C$ into a single vertex
    \STATE $R:=R\cup E(C)$
  \ENDWHILE
\end{algorithmic}
\end{algorithm}

\begin{lemma}\label{lem:2ECS-Construction}
Every 2-edge-connected multigraph has a cycle with an internal cut.
\end{lemma}
\begin{proof}
Let $N(v)$ denote the neighbors of a vertex $v$. We give a direct construction for a cycle $C$ with a vertex $\bar v$ such that all neighbors of $\bar v$ are in the cycle. Then the cut $\{\bar v\}$ will be internal to this cycle. Our construction maintains a path $P$ and repeatedly updates a vertex $\bar v$ to be the last vertex of the path as it grows. 

\begin{algorithmic}[1]
  \STATE Set $P$ to any edge in $G$
  \STATE Set $\bar v$ to be the last vertex in the path $P$
  \WHILE{$\exists u\in N(\bar v) \setminus V(P)$}
    \STATE $P:=P$ concatenated with the edge $\bar vu$
    \STATE $\bar v:=u$
  \ENDWHILE
  \STATE Set $w$ to be the vertex in $N(\bar v)$ earliest in $P$ 
  \STATE Set $C$ to the cycle using $\bar vw$ and edges in $P$
\end{algorithmic}
Our choice of $C$ immediately gives us that $N(\bar v)\subset V(C)$, which implies that the cut $\{\bar v\}$ is internal to $C$.
 \end{proof}

Let $n$ denote the number of vertices and $k$ denote the number of cycles contracted by Algorithm~\ref{alg:1.5-2ECS}. Then we bound the optimal objective value (denoted by $|OPT|$) and the objective value of this algorithm's output (denoted by $|R|$) as follows:
\begin{lemma}
$|OPT| \geq \max\{n,2k\}$
\end{lemma}
\begin{proof}
Consider the dual solution given by combining the internal cuts in each cycle. This is feasible since any edge crossed by one of these internal cuts is removed from the graph in the following contraction. Thus we have a dual feasible solution with objective $2k$. Further any {\sc 2ECS} solution must have at least $n$ edges. Then the cost of the optimal solution is at least $\max\{n, 2k\}$.  
\end{proof}
\begin{lemma}
$|R| = n+k-1 $
\end{lemma}
\begin{proof}
Let $C_i$ be the number of cycles of size $i$ contracted by our algorithm. Since each cycle of size $i$ reduces the number of vertices by $i-1$ and our final graph has a single vertex, we know $\sum^{n}_{i=2} (i-1)C_i=n-1$. Then our algorithm's output has cost $\sum^{n}_{i=2} iC_i = n+k-1$.
\end{proof}
Simple algebra can show $\frac{n+k-1}{max\{n,2k\}}< \frac{3}{2}$. Thus this algorithm has a 1.5-approximation ratio.

\section{A 1.5-Approximation for {\sc DPA}}
\label{sec:1.5DPA}

To apply our dual-based methodology to {\sc DPA}, we need to formulate it as a cut-based integer linear program. Since {\sc DPA} has two types of arcs (high and low power), the corresponding program has to distinguish between these. The program corresponding to {\sc SSC} avoids having different types of arcs and thus has a simpler form.
Then for ease of notation, we will give our dual-based algorithm for {\sc DPA} by approximating an instance of {\sc SSC} with a bidirected input digraph. In Lemma~\ref{lem:DPATOSSC}, we claimed these problems are equivalent and problem instances can be easily transformed between the two. We now prove this result.

\begin{proof}\emph{of Lemma~\ref{lem:DPATOSSC}. }
We give a procedures that will turn any instance of {\sc DPA} on digraph $H$ into an instance of {\sc SSC} with a bidirected input digraph, $(G,\mathcal{C})$, and the reverse direction. Our transformations have linear runtime and will not substantially increase the size of the problem instance. Then our equivalence will follow.

We first consider transforming an instance of {\sc DPA} into {\sc SSC}. Let $H_0$ be the digraph induced by assigning no vertices of $H$ high power. Then $H_0$ will only have zero cost arcs. Since instances of {\sc DPA} are bidirected, no arcs cross between the strongly connected components of $H_0$.
We then construct an instance of {\sc SSC} with a vertex for each strongly connected component of $H_0$. For each vertex $v$ in $H_0$, we add a star with source at the strong component of $v$ and arcs going to each other strong component that $v$ has an arc to in $H$. Note that $G$ is bidirected.
Feasible solutions to these {\sc DPA} and {\sc SSC} instances can be converted between the two while preserving objective as follows:
Given a feasible solution to {\sc DPA}, for each vertex assigned high power add the corresponding star to our {\sc SSC} solution. Similarly, given a feasible solution to {\sc SSC}, we assign each vertex to high power when the corresponding star is in our {\sc SSC} solution. This will produce a feasible {\sc DPA} instance. Note these conversions will have equal objective since there is a one-to-one mapping between high power vertices and stars.

Now we give a transformation for the reverse direction from an instance of {\sc SSC} with a bidirected input digraph, $(G,\mathcal{C})$. Our instance of {\sc DPA} will have a vertex $v_F$ for every star $F\in \mathcal{C}$.
For each vertex $v$ of $G$, consider the set of stars sourced at $v$, $\{F |source(F)=1\}$. Add zero-cost arcs forming a cycle over this set.
For every pair of arcs, $uv$ and $vu$, in our bidirected $G$ and for every star $F$ with $uv\in F$ and star $F'$ with $vu\in F'$, we add a one-cost arc between to $v_F$ and $v_{F'}$. The resulting digraph will be bidirected, as is required for {\sc DPA}. As with our previous transformation, there is a one-to-one relationship between high power vertices and stars. This relationship immediately gives a conversion between our feasible solutions that will maintain objective.
 \end{proof}

We now proceed to construct a cut-based program for {\sc SSC} and then give all the relevant definitions for our algorithms. Our approximations for both {\sc DPA} and {\sc SSC} will utilize these definitions.
For the remainder of our definitions, we consider an instance of {\sc SSC} on a digraph $G=(V,E)$ and a set of stars $\mathcal{C}$.

\begin{definition}
For any star $F\in \mathcal{C}$, we define $source(F)$ to be the common source vertex of all arcs in $F$. We define $sinks(F)$ to be the set of endpoints of arcs in $F$.
\end{definition}

\begin{definition}
For a cut, $\emptyset \subset S \subset V$, we define $\partial \mathcal{C}(S)$ to be the set of all $F\in \mathcal{C}$ such that $source(F)\in S$ and at least one element of $sinks(F)$ is in $V\setminus S$.
\end{definition}

This notation allows us to use $\partial \mathcal{C}(S)$ as the set of all stars with an arc crossing from $S$ to $V\setminus S$. Using these definitions, we can create a cut-based linear programming relaxation for {\sc SSC} similar to those proposed in~\cite{Polzin2003} and~\cite{Calinescu2012}. 

\begin{alignat*}{2}
	\text{{\sc SSC Primal LP}} \\
	\text{minimize }   & \sum_{F\in \mathcal{C}} x_F \\
    \text{subject to } & \sum_{F\in \partial \mathcal{C}(S)}x_F\geq 1 \ & , \ & \bigforall \emptyset\subset S\subset V\\
                       & x_F\geq 0 \ & , \ & \bigforall F\in \mathcal{C}
\end{alignat*}

\begin{lemma} \label{lem:SSC-LPcorrectness}
When {\sc SSC Primal LP} is restricted to $x_F\in \mathbb{Z}$, it is exactly {\sc SSC}.
\end{lemma}

We defer the proof of Lemma~\ref{lem:SSC-LPcorrectness} to the appendix. Intuitively, the dual of {\sc SSC} is to find the maximum set of cuts, such that no star $F\in \mathcal{C}$ crosses multiple of our cuts. Properly, we can consider fractional cuts in our dual problem, but our algorithm only uses integer solutions to the dual problem. 

\begin{alignat*}{2}
	\text{{\sc SSC Dual LP}} \\
	\text{maximize }   & \sum_{\emptyset\subset S\subset V} y_S \\
    \text{subject to } & \sum_{F\in \partial \mathcal{C}(S)}y_S\leq 1 \ & , \ & \bigforall F\in \mathcal{C}\\
                       & y_S\geq 0 \ & , \ & \bigforall \emptyset\subset S\subset V
\end{alignat*}


\subsection{Definitions}
\label{sec:1.5DPA-Defintitons}

In our approximation of {\sc 2ECS}, we repeatedly found cycles in the graph and contracted them. A cycle of length $k$, adds $k$ to the cost of the solution and reduces the number of vertices by $k-1$. A similar method has been applied to {\sc MSCS} in many previous works.
In~\cite{Calinescu2014}, C{\u{a}}linescu proposed a novel way to extend this approach to {\sc DPA}. Following from those definitions, we will use the following two definitions to define a contractible structure in {\sc SSC}.

\begin{definition}
A set $Q\subseteq \mathcal{C}$ is \emph{quasiperfect} if and only if all $F\in Q$ have a distinct $source(F)$ and the subgraph with vertex set the sources of the stars of $Q$ and arc set $\bigcup_{F\in Q}F$ is strongly connected. (Here we abuse notation as $\bigcup_{F\in Q}F$ may contain arcs with endpoints outside of the source vertices of $Q$. Such arcs are ignored.)
\end{definition}

We will use $source(Q)$ for a quasiperfect $Q$ to be the set of all source vertices in $Q$. The distinction between $source$ defined on $F\in \mathcal{C}$ and $source$ defined on quasiperfect sets will always be clear from context.

\begin{definition}
A set $Q\subseteq \mathcal{C}$ is \emph{perfect} if and only if $Q$ is quasiperfect and all $F\in Q$ have $sinks(F)\subseteq source(Q)$.
\end{definition}

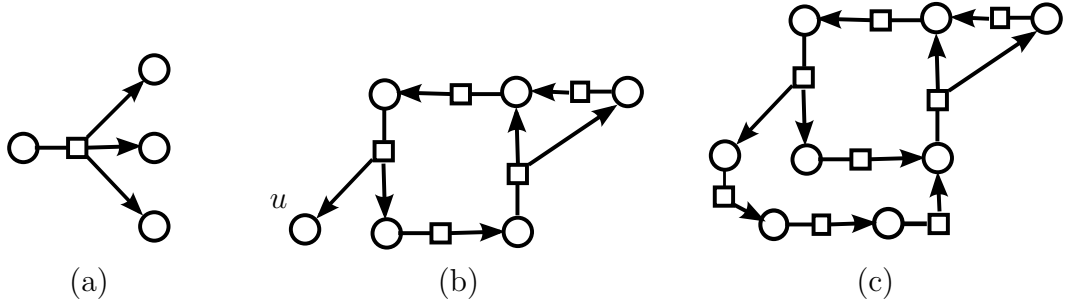
\begin{figure}[t]
\centering
\def\svgwidth{0.85\textwidth} 
\begingroup%
  \makeatletter%
  \providecommand\color[2][]{%
    \errmessage{(Inkscape) Color is used for the text in Inkscape, but the package 'color.sty' is not loaded}%
    \renewcommand\color[2][]{}%
  }%
  \providecommand\transparent[1]{%
    \errmessage{(Inkscape) Transparency is used (non-zero) for the text in Inkscape, but the package 'transparent.sty' is not loaded}%
    \renewcommand\transparent[1]{}%
  }%
  \providecommand\rotatebox[2]{#2}%
  \ifx\svgwidth\undefined%
    \setlength{\unitlength}{1612.327632bp}%
    \ifx\svgscale\undefined%
      \relax%
    \else%
      \setlength{\unitlength}{\unitlength * \real{\svgscale}}%
    \fi%
  \else%
    \setlength{\unitlength}{\svgwidth}%
  \fi%
  \global\let\svgwidth\undefined%
  \global\let\svgscale\undefined%
  \makeatother%
  \begin{picture}(1,0.32678355)%
    \put(0,0){\includegraphics[width=\unitlength]{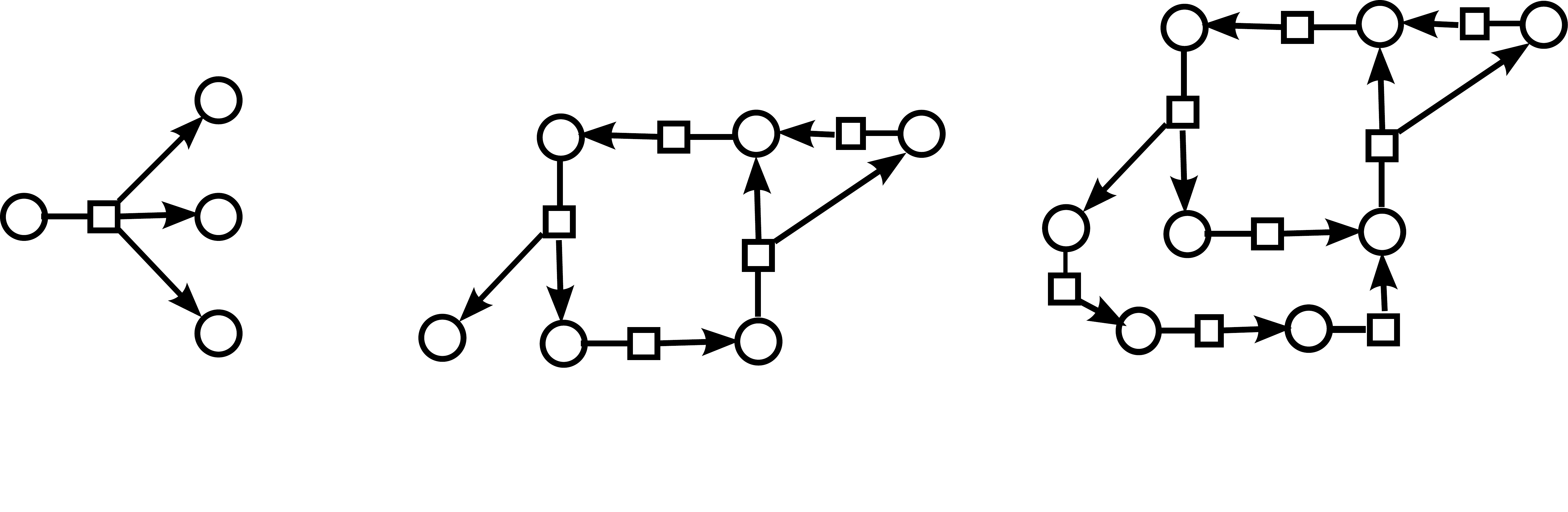}}%
    \put(0.05879989,0.05309538){\color[rgb]{0,0,0}\makebox(0,0)[lb]{\smash{(a)}}}%
    \put(0.4080897,0.05309538){\color[rgb]{0,0,0}\makebox(0,0)[lb]{\smash{(b)}}}%
    \put(0.2487648,0.13184133){\color[rgb]{0,0,0}\makebox(0,0)[lb]{\smash{$u$}}}%
    \put(0.80450548,0.05309538){\color[rgb]{0,0,0}\makebox(0,0)[lb]{\smash{(c)}}}%
  \end{picture}%
\endgroup%
\caption{Examples of both quasiperfect and perfect sets. A star is denoted by a square with a set of arcs leaving it. The vertex connected to the square by a regular line segment is the source. (a) A star made of three arcs. (b) A quasiperfect set of size five. If the arc to vertex $u$ did not exist, this would be a perfect set. (c) A perfect set of size eight containing the previous example. }
\label{fig:StarExamples}
\end{figure}

We define contracting a perfect set as follows: replace all the source vertices of the perfect set with a single supervertex whose arc set is all arcs with exactly one end point in our perfect set.  We can combine duplicate arcs into a single arc during this contraction process. As a result of contraction, the size of a star may decrease, and a star may be removed if it has no remaining arcs.
A quasiperfect $Q$ adds $|Q|$ cost and contracts the $|Q|$ vertices of $source(Q)$ into one, but may have extra arcs leaving the new supervertex. A perfect set has no such arcs, so the problem after contracting such a set will be another instance of {\sc SSC}. Our next lemma describes how to expand any quasiperfect set into a perfect set. This an extension of Lemma 2 given by Calinescu in~\cite{Calinescu2014}.

\begin{lemma} \label{lem:augment}
Every quasiperfect set is a subset of some perfect set.
\end{lemma}
\begin{proof}
Consider the following expansion procedure for any quasiperfect set $Q$.

\begin{algorithmic}[1]
  \WHILE{$\exists F\in Q$ with $u\in sinks(F)\setminus source(Q)$}
    \STATE Find a path $P$ from $u$ to $source(Q)$ that is internally disjoint from $Q$
    \STATE {\bf for each} arc $e$ in $P$ {\bf do} add a star containing $e$ to $Q$ {\bf end for}
  \ENDWHILE
\end{algorithmic}

Line 2 of this construction can be implemented using a simple depth first search. Any star added must have had source outside of $source(Q)$. So no star added will share a source vertex with any other star in $Q$. Further $Q$ will still have a strongly connected subgraph.
Thus each iteration of this procedure maintains the invariant that $Q$ is quasiperfect. When this construction terminates, no such $F$ exists. Therefore the resulting set $Q$ is perfect.
Each iteration also increases the size of $Q$, so it will terminate eventually.
 \end{proof}

Our expansion procedure can be simplified slightly for {\sc DPA}. Since the digraph is bidirected, lines 2 and 3 can just choose any star containing the reverse arc from $u$ to $source(F)$.
For the special case of {\sc MSCS}, all stars have size exactly one. It follows that all quasiperfect sets will be perfect. In fact, for {\sc MSCS}, it can easily be shown that all quasiperfect and perfect sets are cycles.

Our approximation algorithms will repeatedly find perfect sets and contract them. As in our {\sc 2ECS} example, the dual problem requires us to build a set of cuts that share no crossing stars.  We use the following two definitions to relate the dual to perfect sets.

\begin{definition}
Two cuts $S_1$ and $S_2$ are \emph{star-disjoint} if and only if $\partial \mathcal{C}(S_1)$ and $\partial \mathcal{C}(S_2)$ are disjoint (i.e.\ $\partial \mathcal{C}(S_1) \cap \partial \mathcal{C}(S_2)=\emptyset$).
\end{definition}

\begin{definition}
A cut $S$ is \emph{internal} to a set $Q\subseteq \mathcal{C}$ if and only if every $F \in \partial \mathcal{C}(S)$ has $source(F)\in source(Q)$ and $sinks(F)\subseteq source(Q)$.
\end{definition}
Then contracting a perfect set with internal cut $S$ will remove all stars in $\partial \mathcal{C}(S)$ from the digraph. Then $S$ must be star-disjoint from all cuts in the resulting digraph.

\subsection{The 1.5-Approximation Algorithm}
\label{sec:1.5DPA-Algorithm}

Now we will give our approximation algorithm for {\sc DPA} by considering any instance of {\sc SSC} with a bidirected input digraph. In Lemma~\ref{lem:DPA-Construction}, we give a construction for a perfect set with two star-disjoint internal cuts.
Utilizing this lemma, our approximation algorithm becomes very simple.
Our algorithm will repeatedly apply this construction and contract the resulting perfect set. This procedure is formally given in Algorithm~\ref{alg:1.5DPA}.

\begin{algorithm}
\caption{Dual-Based Approximation for {\sc Dual Power Assignment}} \label{alg:1.5DPA}
\begin{algorithmic}[1]
  \STATE $R=\emptyset$
  \WHILE{$|V|\neq 1$}
    \STATE Find a perfect set $Q$ with two star-disjoint internal cuts as shown in Lemma~\ref{lem:DPA-Construction}
	\STATE Contract the sources of $Q$ into a single vertex
    \STATE $R:=R\cup Q$
  \ENDWHILE
\end{algorithmic}
\end{algorithm}

\begin{lemma}\label{lem:DPA-Construction}
Every instance of {\sc SSC} with a bidirected input digraph has a perfect set with two star-disjoint internal cuts.
\end{lemma}
\begin{proof}
We consider an instance of {\sc SSC} defined on a bidirected digraph $G=(V,E)$. In the degenerate case, we have a digraph with only two vertices, $u$ and $v$. Then the perfect set $\{\{uv\},\{vu\}\}$ will have star-disjoint internal cuts $\{u\}$ and $\{v\}$.

Now we assume $|V|\geq 3$. 
Again, we let $N(v)$ denote the neighbors of a vertex $v$. Note that the set of in-neighbors and out-neighbors for a vertex are identical since the digraph is bidirected. We call any vertex with exactly one neighbor a \emph{leaf}. Then our assumption that $|V|\geq 3$ implies some non-leaf vertex exists.
Consider the following cycle construction in $G$ (Figure~\ref{fig:DPA-Cycle-Example} shows its possible outputs).
\begin{algorithmic}[1]
  \STATE Set $P$ to be any arc $r\bar v\in E$, where $\bar v$ is not a leaf 
  \WHILE{TRUE}
    \IF{$\exists u\in V \text{ s.t. } u\in N(\bar v) \setminus V(P)$ and $u$ is not a leaf}
      \STATE $P:=P$ concatenated with the arc $\bar vu$
      \STATE $\bar v:=u$
    \ELSE
      \STATE Set $w$ to be the vertex in $N(\bar v)\cap V(P)$ earliest in $P$ 
      \STATE Set $\bar w$ to be the successor of $w$ in $P$
      \IF{$\exists u\in V \text{ s.t. } u\in N(\bar w) \setminus V(P)$ and $u$ is not a leaf}
        \STATE Replace $P$ with the path using $w\bar v$ instead of $w\bar w$, reversing all arcs between $\bar v$ and $\bar w$
        \STATE $P:=P$ concatenated with the arc $\bar wu$
        \STATE $\bar v:=u$
      \ELSE
        \STATE Set $x$ to be the vertex in $N(\bar w)\cap V(P)$ earliest in $P$ 
        \STATE Set $C$ to the cycle using the arc $\bar vw$, the reverse of arcs in $P$ from $w$ to $x$, the arc $x\bar w$, and arcs in $P$ from $\bar w$ to $\bar v$
        \RETURN $C$, $\bar v$, $\bar w$
      \ENDIF
    \ENDIF    
  \ENDWHILE
\end{algorithmic}

\begin{figure}[t]
\centering
\def\svgwidth{0.65\textwidth} 
\begingroup%
  \makeatletter%
  \providecommand\color[2][]{%
    \errmessage{(Inkscape) Color is used for the text in Inkscape, but the package 'color.sty' is not loaded}%
    \renewcommand\color[2][]{}%
  }%
  \providecommand\transparent[1]{%
    \errmessage{(Inkscape) Transparency is used (non-zero) for the text in Inkscape, but the package 'transparent.sty' is not loaded}%
    \renewcommand\transparent[1]{}%
  }%
  \providecommand\rotatebox[2]{#2}%
  \ifx\svgwidth\undefined%
    \setlength{\unitlength}{1287.68608398bp}%
    \ifx\svgscale\undefined%
      \relax%
    \else%
      \setlength{\unitlength}{\unitlength * \real{\svgscale}}%
    \fi%
  \else%
    \setlength{\unitlength}{\svgwidth}%
  \fi%
  \global\let\svgwidth\undefined%
  \global\let\svgscale\undefined%
  \makeatother%
  \begin{picture}(1,0.46974939)%
    \put(0,0){\includegraphics[width=\unitlength]{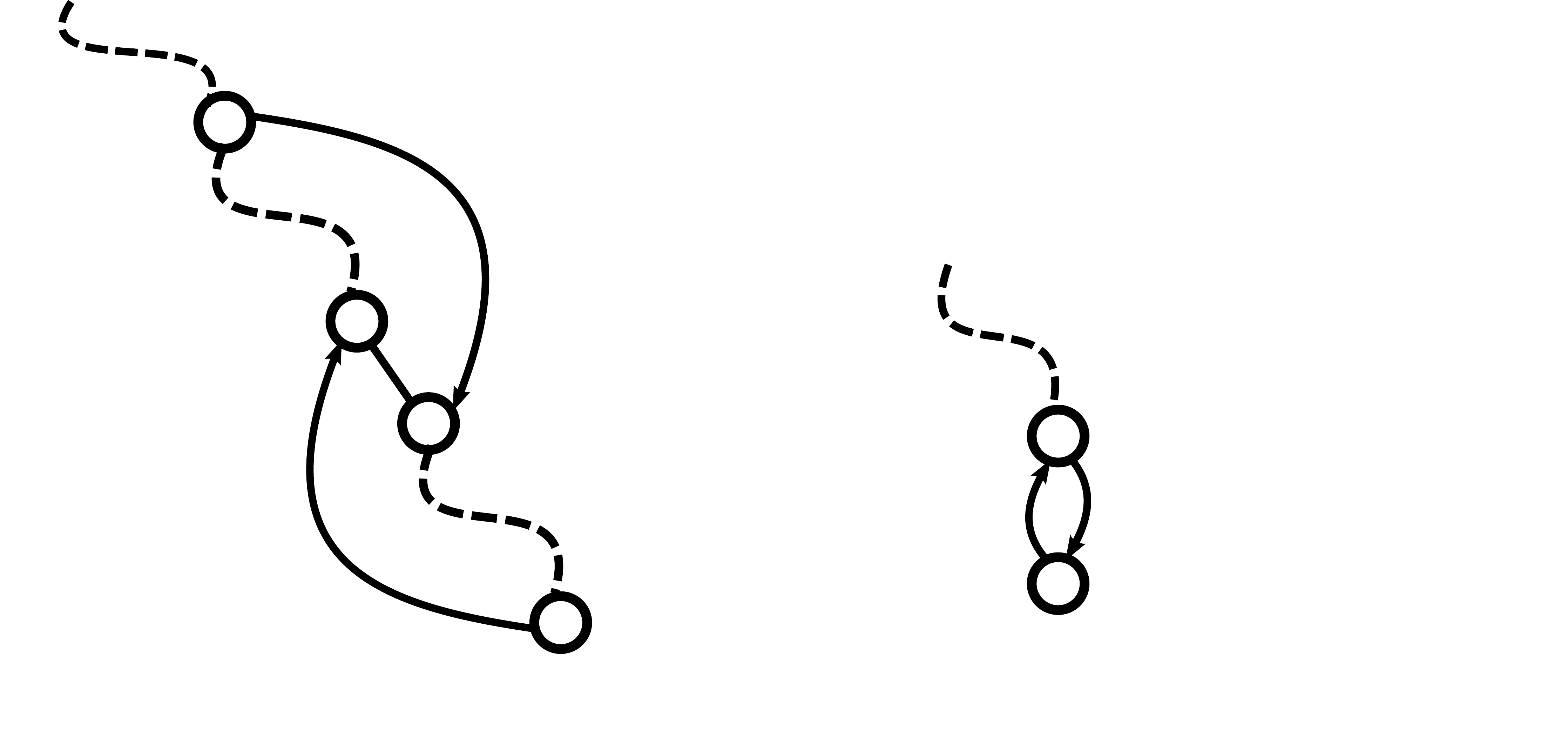}}%
    \put(0.15932342,0.41314417){\color[rgb]{0,0,0}\makebox(0,0)[lb]{\smash{$x$}}}%
    \put(0.23653832,0.28693773){\color[rgb]{0,0,0}\makebox(0,0)[lb]{\smash{$w$}}}%
    \put(0.26156658,0.22676334){\color[rgb]{0,0,0}\makebox(0,0)[lb]{\smash{$\bar w$}}}%
    \put(0.36877995,0.09185914){\color[rgb]{0,0,0}\makebox(0,0)[lb]{\smash{$\bar v$}}}%
    \put(0.69915149,0.0915319){\color[rgb]{0,0,0}\makebox(0,0)[lb]{\smash{$\bar v=\bar w$}}}%
    \put(0.70320042,0.191706){\color[rgb]{0,0,0}\makebox(0,0)[lb]{\smash{$x=w$}}}%
    \put(-0.00309421,0.43089473){\color[rgb]{0,0,0}\makebox(0,0)[lb]{\smash{$p$}}}%
    \put(0.54717303,0.26803337){\color[rgb]{0,0,0}\makebox(0,0)[lb]{\smash{$p$}}}%
    \put(0.23032561,0.00491434){\color[rgb]{0,0,0}\makebox(0,0)[lb]{\smash{(a)}}}%
    \put(0.64959373,0.00491434){\color[rgb]{0,0,0}\makebox(0,0)[lb]{\smash{(b)}}}%
  \end{picture}%
\endgroup%
\caption{Examples of cycles produced by our construction for Lemma~\ref{lem:DPA-Construction}. Dashed curves represent a path. We do not show the leaves that may exist next to $\bar v$ or $\bar w$. (a) Shows the general form of our cycle. (b) Shows the special case when $|C|=2$.}
\label{fig:DPA-Cycle-Example}
\end{figure}
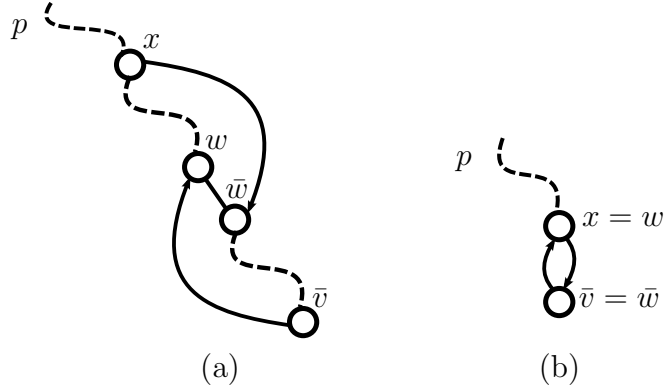

Note that it is possible for $\bar v$ and $\bar w$ to be the same vertex. We make the following claim about the output of this procedure.
\begin{lemma}
For any instance of {\sc SSC} with a bidirected input digraph and $|V|\geq 3$, this construction will output a cycle $C$ with vertices $\bar v, \bar w\in V(C)$ having the following two properties:\vspace{-1.5mm}
\begin{enumerate}
\item $\bar v$ and $\bar w$ are not leaves.
\item Each neighbor of $\bar v$ or $\bar w$ is either in $V(C)$ or a leaf.
\end{enumerate}
\end{lemma}
\begin{proof}
Any strongly connected digraph with at least three vertices will have an initial arc $r\bar v$ where $\bar v$ is not a leaf. This guarantees that step 1 is possible.
Then each iteration increases the length of the path $P$. It follows that there are at most $|V|$ iterations before the construction terminates.

For our first property, we maintain the invariant that $\bar v$ is not a leaf. This is true from our initial choice of $\bar v$, and also maintained in each $u\in V$ chosen to extend $P$. Finally, $\bar w$ is either equal to $\bar v$ and thus not a leaf, or inside the path $P$ and thus has two neighbors.

For our second property, the choice of the cycle $C$ implies that $V(C)$ contains all vertices in $P$ between $x$ and $\bar v$. All non-leaf neighbors of $\bar w$ are at most as early as $x$ in $P$. All non-leaf neighbors of $\bar v$ are at most as early as $w$ in $P$. Note that $w$ is at most as early as $x$ in $p$. Then all the non-leaf neighbors of $\bar v$ and $\bar w$ must be in the V(C). 
 \end{proof}

Using this cycle $C$, we will construct our perfect set with two star-disjoint internal cuts. Note our resulting perfect set may not fully contain $C$.
Let $L_{\bar v}$ and $L_{\bar w}$ be the set of leaves adjacent to $\bar v$ and $\bar w$, respectively.
Consider the case where there is a star $F$ sourced at $\bar v$ containing arcs to multiple leaves. Let $l_1$ and $l_2$ be two distinct leaves in $sinks(F)$. Then we expand the quasiperfect set $\{F\}$ into a perfect set $Q$ using Lemma~\ref{lem:augment}. The resulting set is shown in Figure~\ref{fig:DPA-Construction_Examples} (a). This $Q$ will have internal cuts $\{l_1\}$ and $\{l_2\}$. Note that if two cuts share no vertices, then they also share no crossing stars (i.e.\ they are star-disjoint).
The same construction can be made for such a star sourced at $\bar w$.
For the remainder of our proof, we can assume no star exists sourced from $\bar v$ or $\bar w$ going to multiple leaves.
Now we consider two separate cases: $\bar v = \bar w$ and $\bar v \neq \bar w$.

\begin{description}
\item[Case 1: $\bar v = \bar w$.] In this case, $w$ must be the predecessor of $\bar v$. We can conclude that $|C|=2$. Further, $\bar v$ is only adjacent to $w$ and leaves.
We know that $\bar v$ is not a leaf. Therefore the set $L_{\bar v}$ must be non-empty. Let $l\in L_{\bar v}$ be a leaf of $\bar v$.

As previously shown, we can assume that no star sourced at $\bar v$ contains multiple of our leaves.  Then all stars containing the arc $\bar vl$ are either $\{\bar vl, \bar vw\}$ or $\{\bar vl\}$.
Suppose the star $F=\{\bar vl, \bar vw\}$ exists. Then we expand the quasiperfect set $\{F\}$ into a perfect set $Q$ using Lemma~\ref{lem:augment}. The resulting perfect set will have $\{l\}$ as an internal cuts since the only neighbor of $l$ is $\bar v$. Further, $V\setminus \{l\}$ is internal to this perfect set since $\{\bar vl, \bar vw\}$ and $\{\bar vl\}$ are the only stars containing this arc and all three vertices $w$, $\bar v$, and $l$ are inside the perfect set. These two internal cuts are star-disjoint since they have no common vertices.

If the star $\{\bar vl, \bar vw\}$ does not exist, then the star $\{\bar vl\}$ must exist. Therefore the set $\{\{\bar vl\},\{l\bar v\}\}$ is perfect and has internal cuts $\{l\}$ and $V\setminus \{l\}$. These two possibilities are shown in Figure~\ref{fig:DPA-Construction_Examples} (b) and (c).

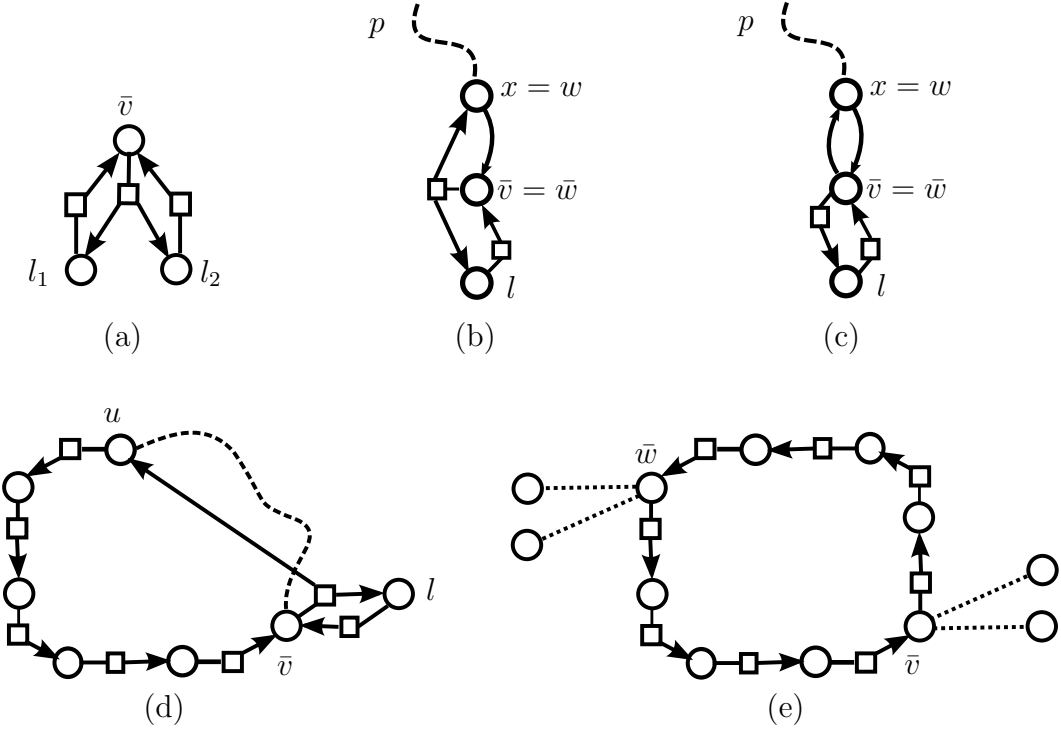
\begin{figure}[t]
\centering
\def\svgwidth{0.9\textwidth} 
\begingroup%
  \makeatletter%
  \providecommand\color[2][]{%
    \errmessage{(Inkscape) Color is used for the text in Inkscape, but the package 'color.sty' is not loaded}%
    \renewcommand\color[2][]{}%
  }%
  \providecommand\transparent[1]{%
    \errmessage{(Inkscape) Transparency is used (non-zero) for the text in Inkscape, but the package 'transparent.sty' is not loaded}%
    \renewcommand\transparent[1]{}%
  }%
  \providecommand\rotatebox[2]{#2}%
  \ifx\svgwidth\undefined%
    \setlength{\unitlength}{1707.7542015bp}%
    \ifx\svgscale\undefined%
      \relax%
    \else%
      \setlength{\unitlength}{\unitlength * \real{\svgscale}}%
    \fi%
  \else%
    \setlength{\unitlength}{\svgwidth}%
  \fi%
  \global\let\svgwidth\undefined%
  \global\let\svgscale\undefined%
  \makeatother%
  \begin{picture}(1,0.64427582)%
    \put(0,0){\includegraphics[width=\unitlength]{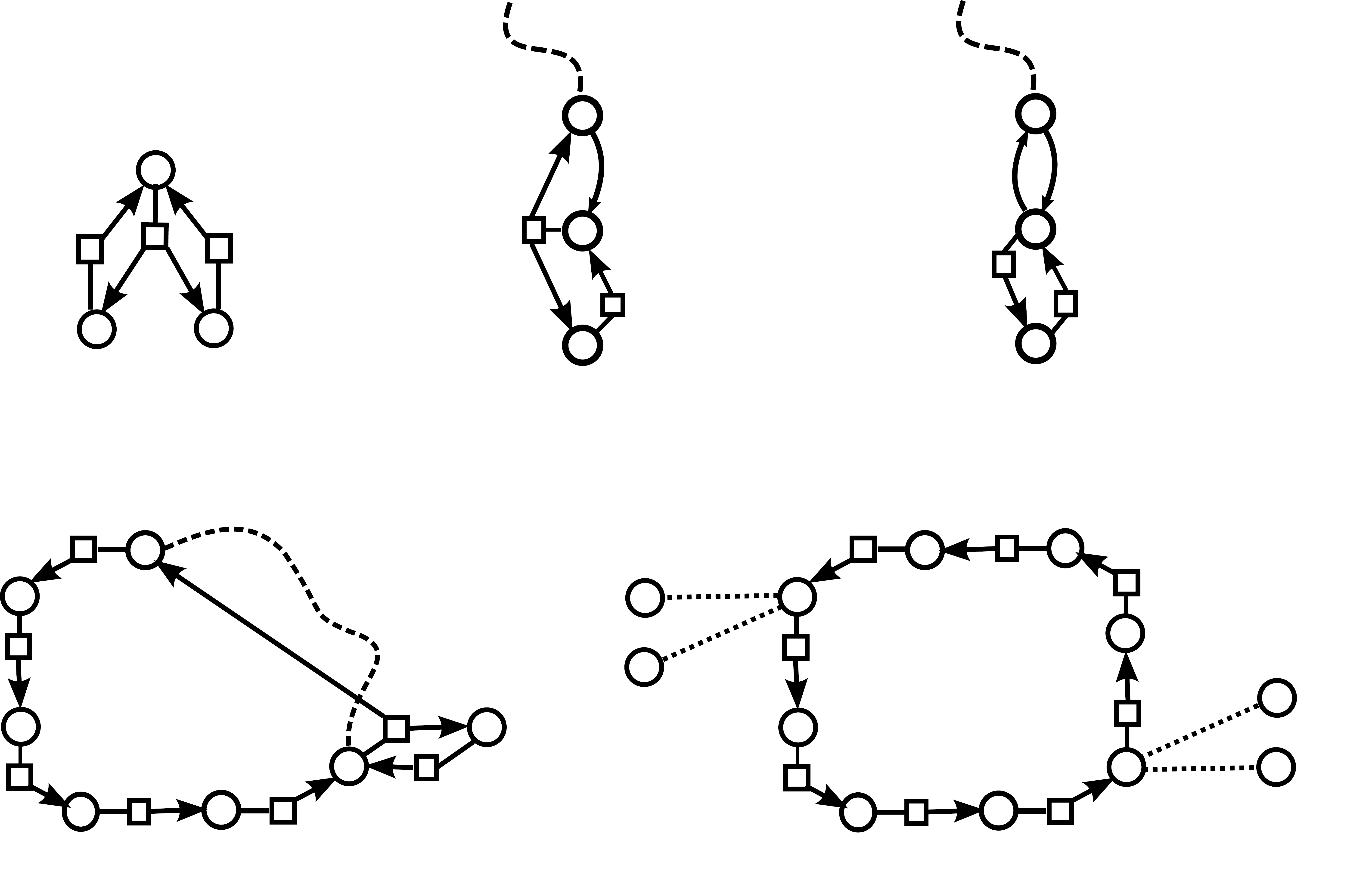}}%
    \put(0.10327497,0.54340373){\color[rgb]{0,0,0}\makebox(0,0)[lb]{\smash{$\bar v$}}}%
    \put(0.02310283,0.39403464){\color[rgb]{0,0,0}\makebox(0,0)[lb]{\smash{$l_1$}}}%
    \put(0.1770226,0.39390078){\color[rgb]{0,0,0}\makebox(0,0)[lb]{\smash{$l_2$}}}%
    \put(0.08995072,0.33655312){\color[rgb]{0,0,0}\makebox(0,0)[lb]{\smash{(a)}}}%
    \put(0.40473242,0.33655312){\color[rgb]{0,0,0}\makebox(0,0)[lb]{\smash{(b)}}}%
    \put(0.44282802,0.46765496){\color[rgb]{0,0,0}\makebox(0,0)[lb]{\smash{$\bar v=\bar w$}}}%
    \put(0.44588096,0.56005281){\color[rgb]{0,0,0}\makebox(0,0)[lb]{\smash{$x=w$}}}%
    \put(0.32823279,0.61760541){\color[rgb]{0,0,0}\makebox(0,0)[lb]{\smash{$p$}}}%
    \put(0.45117317,0.37921803){\color[rgb]{0,0,0}\makebox(0,0)[lb]{\smash{$l$}}}%
    \put(0.77315328,0.4689571){\color[rgb]{0,0,0}\makebox(0,0)[lb]{\smash{$\bar v=\bar w$}}}%
    \put(0.77620622,0.56135495){\color[rgb]{0,0,0}\makebox(0,0)[lb]{\smash{$x=w$}}}%
    \put(0.65855799,0.61890755){\color[rgb]{0,0,0}\makebox(0,0)[lb]{\smash{$p$}}}%
    \put(0.78243539,0.38052017){\color[rgb]{0,0,0}\makebox(0,0)[lb]{\smash{$l$}}}%
    \put(0.73387232,0.33655312){\color[rgb]{0,0,0}\makebox(0,0)[lb]{\smash{(c)}}}%
    \put(0.2460267,0.04188249){\color[rgb]{0,0,0}\makebox(0,0)[lb]{\smash{$\bar v$}}}%
    \put(0.09076848,0.26807768){\color[rgb]{0,0,0}\makebox(0,0)[lb]{\smash{$u$}}}%
    \put(0.37937955,0.10898714){\color[rgb]{0,0,0}\makebox(0,0)[lb]{\smash{$l$}}}%
    \put(0.12654592,0.00395256){\color[rgb]{0,0,0}\makebox(0,0)[lb]{\smash{(d)}}}%
    \put(0.80726504,0.04271086){\color[rgb]{0,0,0}\makebox(0,0)[lb]{\smash{$\bar v$}}}%
    \put(0.56596597,0.23142988){\color[rgb]{0,0,0}\makebox(0,0)[lb]{\smash{$\bar w$}}}%
    \put(0.68400317,0.00395256){\color[rgb]{0,0,0}\makebox(0,0)[lb]{\smash{(e)}}}%
  \end{picture}%
\endgroup%
\caption{Depicting the quasiperfect or perfect sets selected by our construction based on the cycle $C$ with vertices $\bar v$ and $\bar w\in V(C)$. We omit additional arcs that could be part of each star shown. If such arcs exist, they will be handled by the expansion shown in Lemma~\ref{lem:augment}. Note $\bar v$ and $\bar w$ are symmetric and each case shown applies to both.}
\label{fig:DPA-Construction_Examples}
\end{figure}

\item[Case 2: $\bar v \neq \bar w$.] First suppose a star sourced at $\bar v$ exists with an arc to a leaf $l\in L_{\bar v}$ and an arc to a cycle vertex $u\in V(C)$. Then let $F$ be the star of this form with $u$ nearest after $\bar v$ in $C$ (using an arbitrary direction around $C$). We construct a quasiperfect set by taking $F$ and a star containing each arc of the cycle from $u$ to $\bar v$. We can expand this quasiperfect set into a perfect set $Q$ using Lemma~\ref{lem:augment}. Then $\{l\}$ is an internal cut to $Q$. This perfect set is shown in Figure~\ref{fig:DPA-Construction_Examples} (d).
We know that any star $F$ sourced at $\bar v$ with an arc to $l$ cannot have an arc to any other leaf. Then our choice of $Q$ gives us that all stars $F$ crossing the cut $V\setminus \{l\}$ have $sinks(F) \subset source(Q)$. Therefore the cut $V\setminus \{l\}$ is also internal to $Q$.

Now, we can assume no star sourced at $\bar v$ exists with arcs into both $L_{\bar v}$ and $V(C)$. Then every star $F$ crossing the cut $\{\bar v\}\cup L_{\bar v}$ has $sinks(F)\subset V(C)$. The same claim holds for $\bar w$ by symmetry.
Then we construct a quasiperfect set by iterating over the arcs of $C$, selecting a star containing each arc. Let $Q$ be the perfect set made by expanding this quasiperfect set using Lemma~\ref{lem:augment}. This perfect set is shown in Figure~\ref{fig:DPA-Construction_Examples} (e).
We will have $\{\bar v\}\cup L_{\bar v}$ as an internal cut since $V(C) \subseteq source(Q)$. Similarly we also have the internal cut $\{\bar w\}\cup L_{\bar w}$. Since $\bar v \neq \bar w$, the two internal cuts are star-disjoint.

\end{description}

Therefore in either of our cases we can construct a perfect set with two star-disjoint internal cuts. This concludes our proof of Lemma~\ref{lem:DPA-Construction}.
 \end{proof}

\subsection{Analysis of 1.5-Approximation Ratio}
\label{sec:1.5DPA-Analysis}

Let $\mathcal{I}$ be the set of all possible {\sc SSC} problem instances with bidirected input digraphs. We denote the solution from our algorithm on some $I\in \mathcal{I}$ as $A(I)$ and the optimal solution as $OPT(I)$. We let $k$ denote the total number of perfect sets added by our algorithm. Further, we let $A_i$ denote the number of perfect sets of size $i$ added by the algorithm. 

\begin{lemma}\label{lem:DPA-A(G)Bound}
$|A(I)|=n+k-1$
\end{lemma}
\begin{proof}
Observe that $A_i=0$ for any $i>n$ since the source of stars in a perfect set are distinct. Each perfect set of size $i$ contracts $i-1$ vertices, and over the whole algorithm, we contract $n$ vertices into 1. Therefore $\sum_{i=2}^n (i-1)A_i=n-1$.
Each perfect set of size $i$ contributes $i$ cost to our solution. Then our cost is $\sum_{i=2}^n iA_i=n+k-1$.
 \end{proof}

\begin{lemma}\label{lem:DPA-DualBound1}
$|OPT(I)|\geq n$
\end{lemma}
\begin{proof}
Consider the dual solution of assigning one to the cut $\{v\}$ for all $v\in V$. This dual feasible solution has objective $n$. The lemma follows from weak duality.
 \end{proof}

\begin{lemma}\label{lem:DPA-DualBound2}
$|OPT(I)|\geq 2k$
\end{lemma}
\begin{proof}
Whenever the algorithm adds a perfect set, we can identify two star-disjoint internal cuts. We construct a dual feasible solution by assigning each of our internal cuts $y_S=1$. To accomplish this, we must show that no star $F$ crosses multiple of our internal cuts. Consider any star $F\in \mathcal{C}$ crossed by at least of our cuts. Let $S$ be the first of our internal cuts with $F\in \partial \mathcal{C}(S)$. Since $S$ is an internal cut, all arcs of $F$ will be removed from the graph after contracting the corresponding perfect set. Thus this $F$ will be crossing at most one of our internal cuts.
So we have a dual feasible solution with objective $2k$. The lemma follows from weak duality.
 \end{proof}

By taking a convex combination of Lemmas~\ref{lem:DPA-DualBound1} and~\ref{lem:DPA-DualBound2}, we know the following:

\begin{equation} \label{eq:DPA-ConvexCombination}
|OPT(I)| \geq \frac{2}{3}n + \frac{1}{3}(2k).
\end{equation}

Taking the ratio of Lemma~\ref{lem:DPA-A(G)Bound} and Equation (\ref{eq:DPA-ConvexCombination}), we get a bound on the approximation ratio. Straightforward algebra on this ratio completes our proof of Theorem~\ref{thm:1.5DPA}:
\begin{align*}
\frac{|A(I)|}{|OPT(I)|} & \leq \frac{n+k-1}{\frac{2n}{3}+\frac{2k}{3}} < \frac{3}{2}=1.5.
\end{align*}

As a corollary, 1.5 upper bounds the ratio between integer primal and integer dual solutions to our program. This is easily verified on any {\sc SSC} instance $I$ by choosing $A(I)$ for the integer primal and the larger of the two dual solutions used in Lemma~\ref{lem:DPA-DualBound1} and~\ref{lem:DPA-DualBound2}. Corollary~\ref{cor:1.5DPA-IntegralityUpper} follows from this observation.
Further, our analysis of the approximation ratio is tight as shown by an example in our appendix.
\begin{theorem} \label{thm:1.5DPA-Tightness}
The 1.5-approximation ratio of Algorithm~\ref{alg:1.5DPA} is tight.
\end{theorem}

\section{A 1.6-Approximation for {\sc SSC}}
\label{sec:1.6SSC}

As in our approximation for {\sc DPA}, we need a method to construct perfect sets with internal cuts. Without the restriction to bidirected input digraphs, we are unable the guarantee two star-disjoint internal cuts in each perfect set. Such a construction would give {\sc SSC} a 1.5-approximation. Instead, we guarantee the following weaker condition.

\begin{lemma}\label{lem:1.6SSC-Construction}
Every instance of {\sc SSC} has a perfect set $Q$ with either $|Q|\geq 4$ and one internal cut, or two star-disjoint internal cuts. 
\end{lemma}
\begin{proof}
We use $N^+(v)$ to denote the set of out-neighbors of a vertex $v$ in $G$. We use the following cycle construction, which is a simplification of the construction used for {\sc DPA}.

\begin{algorithmic}[1]
  \STATE Set $P$ to be any arc $r\bar v\in E$
  \WHILE{$\exists u\in N^+(\bar v) \setminus V(P)$}
    \STATE $P:=P$ concatenated with the arc $\bar vu$
    \STATE $\bar v:=u$
  \ENDWHILE
  \STATE Set $w$ to be the vertex in $N^+(\bar v)$ earliest in $P$ 
  \STATE Set $C$ to the cycle using $\bar vw$ and arcs in $P$
  \RETURN $C$ and $\bar v$
\end{algorithmic}

\begin{lemma} \label{lem:cycle}
This construction will output a cycle $C$ and $\bar v \in V(C)$ such that $N^+(\bar v)\subseteq V(C)$.
\end{lemma}
\begin{proof}
Follows immediately from our choice of $C$.
 \end{proof}
Using this $C$ and $\bar v$, we will construct our perfect set. We consider this in three separate cases: $|C|\geq 4$, $|C|=3$, and $|C|=2$.

\begin{description}
\item[Case 1: $|C|\geq 4$.] 
We construct a quasiperfect set by iterating over the arcs of our cycle $C$, selecting a star containing each arc.
Then let $Q$ be the perfect set created by expanding this set using Lemma~\ref{lem:augment}. Note that $|Q|\geq 4$. Since $N^+(\bar v)\subset V(C)$, the cut $\{\bar v\}$ will be internal to $Q$.


\item[Case 2: $|C|=3$.] Our cycle construction must have found a cycle $C$ on vertices $\{\bar v, u_1, u_2\}$. Then $\bar v$ has the property that $N^+(\bar v)\subseteq\{u_1,u_2\}$. Suppose some of the cycle arcs, $\bar vu_1$, $u_1u_2$, or $u_2\bar v$, are part of a star with a vertex outside our cycle in its sink set. Then we can construct a quasiperfect set containing this star and a star for each other arc in the cycle. Expanding this quasiperfect set, as defined in Lemma~\ref{lem:augment}, will produce a perfect set of size four or more with the internal cut $\{\bar v\}$. So we assume that no such $F$ exists, and thus $C$ is a perfect set. 

If there exists a nontrivial path from $u_1$ to $u_2$ or from $u_2$ to $\bar v$ internally disjoint from $V(C)$, then we can replace an arc of $C$ with this path to get a larger cycle that has the same property with respect to $\bar v$ (nontrivial meaning with $|P|\geq2$). Then we can apply Case 1 to handle the new cycle.

Similarly, we can assume at least one of the following does not exist: path from $u_2$ to $u_1$ internally disjoint from $V(C)$, path from $u_1$ to $\bar v$ internally disjoint from $V(C)$, or the arc from $\bar v$ to $u_2$.
If all three of these existed and at least one of the paths is nontrivial, we could construct a larger cycle with the same property with respect to $\bar v$ by starting at $\bar v$, following the arc to $u_2$, following the path to $u_1$, finally following the path to $\bar v$. We know the paths from $u_2$ to $u_1$ and $u_1$ to $\bar v$ are internally disjoint because any overlap would create a path from $u_2$ to $\bar v$.
Thus this construction will produce a larger cycle containing all of neighbors of $\bar v$. Then we can assume this structure does not exist.

We handle the four possible cases of our assumption separately. Note that the strong connectivity of the digraph implies there exists a path between any pair of vertices. Let $R_C(u)$ be the set of vertices reachable by $u$ without using any arcs with both endpoints in $C$.
\begin{description}
  \item[Subcase 2.1] No path from $u_2$ to $u_1$ exists that is internally disjoint from $V(C)$.
  Note that we also assume no nontrivial path from $u_2$ to $\bar v$ exists. Then $R_C(u_2)$ includes neither $\bar v$ nor $u_1$. Further, the only arc crossing $R_C(u_2)$ is $u_2\bar v$ (since we have assumed the arc $u_2u_1$ does not exist). It follows that $R_C(u_2)$ is internal to any perfect set contracting all three cycle vertices since no star contains a cycle arc and an arc to a fourth external vertex.
  
  Furthermore, the cut $\{\bar v\}$ will be internal to any perfect set contracting all three cycle vertices. Therefore we select any star containing each of our cycle arcs to produce a perfect set with internal cuts $\{\bar v\}$ and $R_C(u_2)$. These two cuts are star-disjoint since they have no vertices in common. \\

  \item[Subcase 2.2] No path from $u_1$ to $\bar v$ exists that is internally disjoint from $V(C)$.
  Note that we also assume no nontrivial path from $u_1$ to $u_2$ exists. Then $R_C(u_1)$ includes neither $\bar v$ nor $u_2$. Further, the only arc crossing this cut is $u_1u_2$ (since we have assumed the arc $u_1\bar v$ does not exist). It follows that $R_C(u_1)$ is internal to any perfect set contracting all three cycle vertices since no star contains a cycle arc and an arc to a fourth external vertex.
  
  Furthermore, the cut $\{\bar v\}$ will be internal to any perfect set contracting all three cycle vertices. Therefore we select any star containing each of our cycle arcs to produce a perfect set with internal cuts $\{\bar v\}$ and $R_C(u_1)$. These two cuts are star-disjoint since they have no vertices in common. \\

  \item[Subcase 2.3] The arc $\bar vu_2$ does not exist.
  Again, we select a perfect set made by selecting a star containing each arc of our cycle. The cut $\{\bar v\}$ is internal to $C$ from our choice of $C$ and $\bar v$. Further, we know that no nontrivial path exists from $\bar v$ to $u_2$ internally disjoint from $V(C)$ and no nontrivial path exists from $u_1$ to $u_2$. Then we can conclude that the cut $\{\bar v\}\cup R_C(u_1)$ is only crossed by the arc $u_1u_2$ (since $\bar vu_2$ does not exist). Since we assumed no star containing a cycle arc and an arc to a fourth external vertex exists, $\{\bar v\}\cup R_C(u_1)$  is also internal to our perfect set.
These two cuts are star-disjoint because the arc $\bar vu_2$ does not exist.\\

  \item[Subcase 2.4] The arcs $\bar vu_2$, $u_2u_1$ and $u_1\bar v$ exist, but no nontrivial paths exist from $u_2$ to $u_1$ or from $u_1$ to $\bar v$ internally disjoint from $V(C)$.
Suppose any of the reversed cycle arcs $\bar vu_2$, $u_2u_1$ or $u_1\bar v$ are part of a star containing an arc to a fourth vertex outside of $V(C)$. Then we could construct a perfect set of size four or more with internal cut $\{\bar v\}$ based on this cycle.
  
Now we assume no such stars exist. We select our perfect set by choosing any star containing each of our cycle arcs. The cut $\{\bar v\}$ will be internal to such a perfect set. Since no nontrivial path exists from $u_1$ to $u_2$ or to $\bar v$ that is internally star-disjoint from $V(C)$, the cut $R_C(u_1)$ will only be crossed by the arcs $u_1u_2$ and $u_1\bar v$. Thus $R_C(u_1)$ is internal to our perfect set. These two cuts are star-disjoint since they have no vertices in common.\\
\end{description}

Thus under any case we can find two star-disjoint internal cuts in our cycle or a larger perfect set with one internal cut.

\item[Case 3: $|C|=2$.] Our cycle construction must have found a cycle $C$ on vertices $\{\bar v, u_1\}$. Note that the only arc leaving $\bar v$ goes to $u_1$, and there is a star containing only this arc. The strong connectivity of our digraph implies there is a path from $u_1$ to $\bar v$. If a nontrivial path $Q$ exists from $u_1$ to $\bar v$, then we can replace $C$ with the cycle made by concatenating the arc $\bar vu_1$ with $Q$. This larger cycle can then be processed by either Case 1 or 2. So we can assume that the only path from $u_1$ to $\bar v$ is the arc between them.
Consider the cut $V\setminus \{\bar v\}$, which is only crossed by $u_1\bar v$. Either there exists a $F_1\in C$ containing $u_1\bar v$ and some $u_1u_2$, or this cut is internal to any perfect set contracting $\bar v$ and $u_1$.
In the latter case, we can choose the perfect set $\{\{\bar vu_1\}, \{u_1\bar v\}\}$. This perfect set has two star-disjoint internal cuts:  $\{\bar v\}$ and $V\setminus \{\bar v\}$.

If this $F_1$ and $u_2$ exist, then any nontrivial path $Q$ from $u_2$ to $u_1$ would create a quasiperfect set of size at least four. Then by Lemma~\ref{lem:augment}, we could find a perfect set of size four or more with the internal cut $\{\bar v\}$.
Finally, we handle the case where the only path from $u_2$ to $u_1$ is the arc $u_2u_1$. Let $R\subset V$ be the set of all vertices that can be reached by $u_2$ without using the arc $u_2u_1$.
Consider the cut given by $R$, which is only crossed by $u_2u_1$. Either there exists a $F_2\in C$ containing $u_2u_1$ and some $u_2u_3$, or $R$ is internal to any perfect set contracting $u_1$ and $u_2$.
In the former case, expanding the quasiperfect set $\{F_1,F_2\}$, as defined in Lemma~\ref{lem:augment}, will give a perfect set of size at least four with internal cut $\{\bar v\}$. In the latter case, expanding the quasiperfect set $\{F_1\}$, as defined in Lemma~\ref{lem:augment}, will give a perfect set with two star-disjoint internal cuts: $\{\bar v\}$ and $R$. 

\end{description}

Therefore regardless of the size of $\mathcal{C}$, we can construct either a size four or more perfect set with an internal cut or a smaller perfect set with two star-disjoint internal cuts. This concludes our proof of Lemma~\ref{lem:1.6SSC-Construction}.
 \end{proof}

Using Lemma~\ref{lem:1.6SSC-Construction}, our approximation algorithm is very simple. We repeatedly apply this construction and contract the resulting perfect set. This procedure is formally given in Algorithm~\ref{alg:1.6SSC}.

\begin{algorithm}
\caption{Dual-Based Approximation for {\sc Star Strong Connectivity}} \label{alg:1.6SSC}
\begin{algorithmic}[1]
  \STATE $R=\emptyset$
  \WHILE{$|V|\neq 1$}
    \STATE Find a perfect set $Q$ as shown in Lemma~\ref{lem:1.6SSC-Construction}
	\STATE Contract the sources of $Q$ into a single vertex
    \STATE $R:=R\cup Q$
  \ENDWHILE
\end{algorithmic}
\end{algorithm}

\subsection{Analysis of 1.6-Approximation Ratio }
\label{sec:1.6SSC-Analysis}

The analysis of our approximation ratio is very similar to the analysis given in Section~\ref{sec:1.5DPA-Analysis}. Let $\mathcal{I}$ be the set of all possible {\sc SSC} problem instances. We denote the solution from our algorithm on some $I\in \mathcal{I}$ as $A(I)$ and the optimal solution as $OPT(I)$. We let $A_i$ denote the number of perfect sets of size $i$ added by the algorithm.

\begin{lemma}\label{lem:SSC-A(G)Bound}
$|A(I)|=\sum\limits_{i=2}^n iA_i$
\end{lemma}
\begin{proof}
Each perfect set in $A_i$ has $i$ stars, and thus contributes $i$ cost to our solution.
 \end{proof}

\begin{lemma}\label{lem:SSC-DualBound1}
$|OPT(I)|\geq n>\sum\limits_{i=2}^n (i-1)A_i$
\end{lemma}
\begin{proof}
Consider the dual solution of assigning one to the cut $\{v\}$ for all $v\in V$. This dual feasible solution has objective $n$. Since each perfect set of size $i$ added by the algorithm contracts $i-1$ vertices, we know that $n-1=\sum\limits_{i=2}^n (i-1)A_i$. The lemma follows from weak duality.
 \end{proof}

\begin{lemma}\label{lem:SSC-DualBound2}
$|OPT(I)|\geq 2A_2+2A_3+\sum\limits_{i=4}^n A_i$
\end{lemma}
\begin{proof}
Whenever the algorithm adds a perfect set of size two, we can identify two star-disjoint internal cuts. Similarly, there are two star-disjoint internal cuts in every perfect set of size three and one internal cut in the remaining perfect sets. From the definition of internal cuts, we know that all of these internal cuts will be star-disjoint from previously added internal cuts. So we have a dual feasible solution and the lemma follows from weak duality.
 \end{proof}

By taking a convex combination of Lemmas~\ref{lem:SSC-DualBound1} and~\ref{lem:SSC-DualBound2}, we know the following:

\begin{align}
|OPT(I)|& > \frac{3}{4}(\sum\limits_{i=2}^n (i-1)A_i) + \frac{1}{4}(2A_2+2A_3+\sum\limits_{i=4}^n A_i)\nonumber\\
& =\frac{5}{4}A_2+2A_3+\sum\limits_{i=4}^n (\frac{3}{4}i-\frac{1}{2})A_i. \label{eq:SSC-ConvexCombination}
\end{align}

Taking the ratio of Lemma~\ref{lem:SSC-A(G)Bound} and Equation (\ref{eq:SSC-ConvexCombination}), we get a bound on the approximation ratio. Straightforward algebra on this ratio can show it is at most 8/5:
\begin{align*}
\frac{|A(I)|}{|OPT(I)|} & < \frac{\sum\limits_{i=2}^n iA_i}{\frac{5}{4}A_2+2A_3+\sum\limits_{i=4}^n (\frac{3}{4}i-\frac{1}{2})A_i} \leq \frac{8}{5}=1.6.
\end{align*}

This finishes the proof of Theorem~\ref{thm:1.6SSC}. Corollary~\ref{cor:1.6SSC-IntegralityUpper} follows from the same observation made about our 1.5-approximation for {\sc DPA}. Further, our analysis of the approximation ratio is tight as shown by an example in our appendix.

\begin{theorem}
The 1.6-approximation ratio of Algorithm~\ref{alg:1.6SSC} is tight.
\end{theorem}

\section{Conclusion}
\label{sec:Conclusion}

We introduced a novel approach to approximating network design problems with cut-based linear programming relaxations. Our method combines the combinatorial (recursive) structure of the problem with the cut-based structure produced by the corresponding dual linear program. Identifying subgraphs that meet both the recursive and dual structural constraints can produce provably good approximations.

We applied this methodology to a number of standard network design problems. In the case of {\sc Minimum 2-Edge-Connected Spanning Subgraph}, the resulting algorithm is equivalent to a previously proposed 3/2-approximation~\cite{Khuller1994}. For the problem of {\sc Minimum Strongly Connected Spanning Subgraph}, we produce a tight 1.6-approximation. Although this is slightly worse than the 1.5-approximation of Vetta~\cite{Vetta2001}, our algorithm has notably fewer cases than Vetta's algorithm.

We also applied our dual-based approach to a common power assignment network design problem. The resulting algorithm for {\sc Dual Power Assignment} achieves the best approximation ratio known of 1.5 (improving on the previous best known 1.57-approximation of~\cite{Abu2014}). 
We introduced a new problem generalizing {\sc DPA} and {\sc MSCS}, which we call {\sc Star Strong Connectivity}. Our approach gives a tight 1.6-approximation to {\sc SSC}.
Each of our approximation results also proves an upper bound on the integrality gap of the corresponding problem. 

Our dual-based approach can likely be applied to other unweighted connectivity problems with cut-based linear programs. Further application of this method will likely produce new approximation algorithms and improved upper bounds on their integrality gaps.
Another interesting extension of this work would generalize {\sc SSC} to have costs on each star. Even when costs are constrained to be in $\{0,1\}$, {\sc Weighted SSC} is at least as hard to approximate as {\sc Set Cover}. This follows from a very simple reduction that was observed in~\cite{Calinescu2014}. As a consequence, any work on approximating the weighted variant will at best achieve a logarithmic approximation ratio. We believe the methods used in~\cite{Calinescu2003} will generalize to give {\sc Weighted SSC} such a logarithmic approximation.\\

\paragraph{Acknowledgments.} This research was supported in part by a College of Science Undergraduate Summer Research Award at the Illinois Institute of Technology. We thank Gruia C{\u{a}}linescu for his many valuable comments and fruitful discussions, which notably improved the paper.


\bibliographystyle{plain}
\bibliography{all.bib}

\appendix

\section{Proof of Lemma~\ref{lem:SSC-LPcorrectness}}

Consider some instance of {\sc SSC} given by a digraph $G=(V,E)$ and a set of stars $\mathcal{C}$. Let $R_{OPT}\subseteq C$ be the optimal solution to {\sc SSC}. Let $x^*$ be the optimal solution to {\sc SSC Primal LP} when restricted to $x_F\in \mathbb{Z}$. We then need to show that $|R_{OPT}| = \sum_{F\in \mathcal{C}} x^*_F$.

First we show that $|R_{OPT}| \geq \sum_{F\in \mathcal{C}} x^*_F$. Consider the vector $x$ produced by assigning all $F\in R_{OPT}$ value 1 and the rest value 0. Then $|R_{OPT}| = \sum_{F\in \mathcal{C}} x_F$. Our inequality will follow if we show $x$ is a feasible solution to {\sc SSC Primal LP}, since $x^*$ is the minimum feasible solution. From our construction, all $x_F\geq 0$. Further consider any cut $\emptyset\subset S \subset V$. Since $R_{OPT}$ produces a strongly connected spanning subgraph, some $F\in R_{OPT}$ crosses $S$. Since this $x_F=1$, we know $\sum_{F\in\partial \mathcal{C}(S)}x_F\geq 1$. Thus $x$ is feasible.  

Now we prove that $|R_{OPT}| \leq \sum_{F\in \mathcal{C}} x^*_F$. We know that all $x^*_F\in\{0,1\}$ (if a larger $x^*_F$ exists, our objective is reduced by reducing it to $x^*_F=1$ without effecting feasibility). Consider the set of stars $R=\{F | x^*_F=1\}$. Then $|R| = \sum_{F\in \mathcal{C}} x^*_F$. Our inequality will follow if we show $R$ is a feasible solution to {\sc SSC}, since $R_{OPT}$ is the minimum feasible solution. We prove this by contradiction. Let $G'$ be the digraph induced by $R$ (i.e.\ $G' = (V, \bigcup_{F\in R} F)$). Assume $G'$ is not strongly connected. Then there exists $s,t\in V$ such that there is no $s,t$-path in $G'$. Consider the set $V_s\subseteq V$ of all vertices $u$ with a $s,u$-path. Note $t\notin V_s$. Then $V_s$ is a cut with no arcs or stars crossing it. However, this contradicts the fact that $x^*$ is feasible. Thus we can conclude $R$ is feasible. Lemma~\ref{lem:SSC-LPcorrectness} follows.

\section{Tightness of 1.5-Approximation Ratio for {\sc DPA}}
\label{sec:1.5DPA-Tightness}

\begin{figure}[t]
\centering
\def\svgwidth{0.5\textwidth} 
\begingroup%
  \makeatletter%
  \providecommand\color[2][]{%
    \errmessage{(Inkscape) Color is used for the text in Inkscape, but the package 'color.sty' is not loaded}%
    \renewcommand\color[2][]{}%
  }%
  \providecommand\transparent[1]{%
    \errmessage{(Inkscape) Transparency is used (non-zero) for the text in Inkscape, but the package 'transparent.sty' is not loaded}%
    \renewcommand\transparent[1]{}%
  }%
  \providecommand\rotatebox[2]{#2}%
  \ifx\svgwidth\undefined%
    \setlength{\unitlength}{979.40990454bp}%
    \ifx\svgscale\undefined%
      \relax%
    \else%
      \setlength{\unitlength}{\unitlength * \real{\svgscale}}%
    \fi%
  \else%
    \setlength{\unitlength}{\svgwidth}%
  \fi%
  \global\let\svgwidth\undefined%
  \global\let\svgscale\undefined%
  \makeatother%
  \begin{picture}(1,0.45016995)%
    \put(0,0){\includegraphics[width=\unitlength]{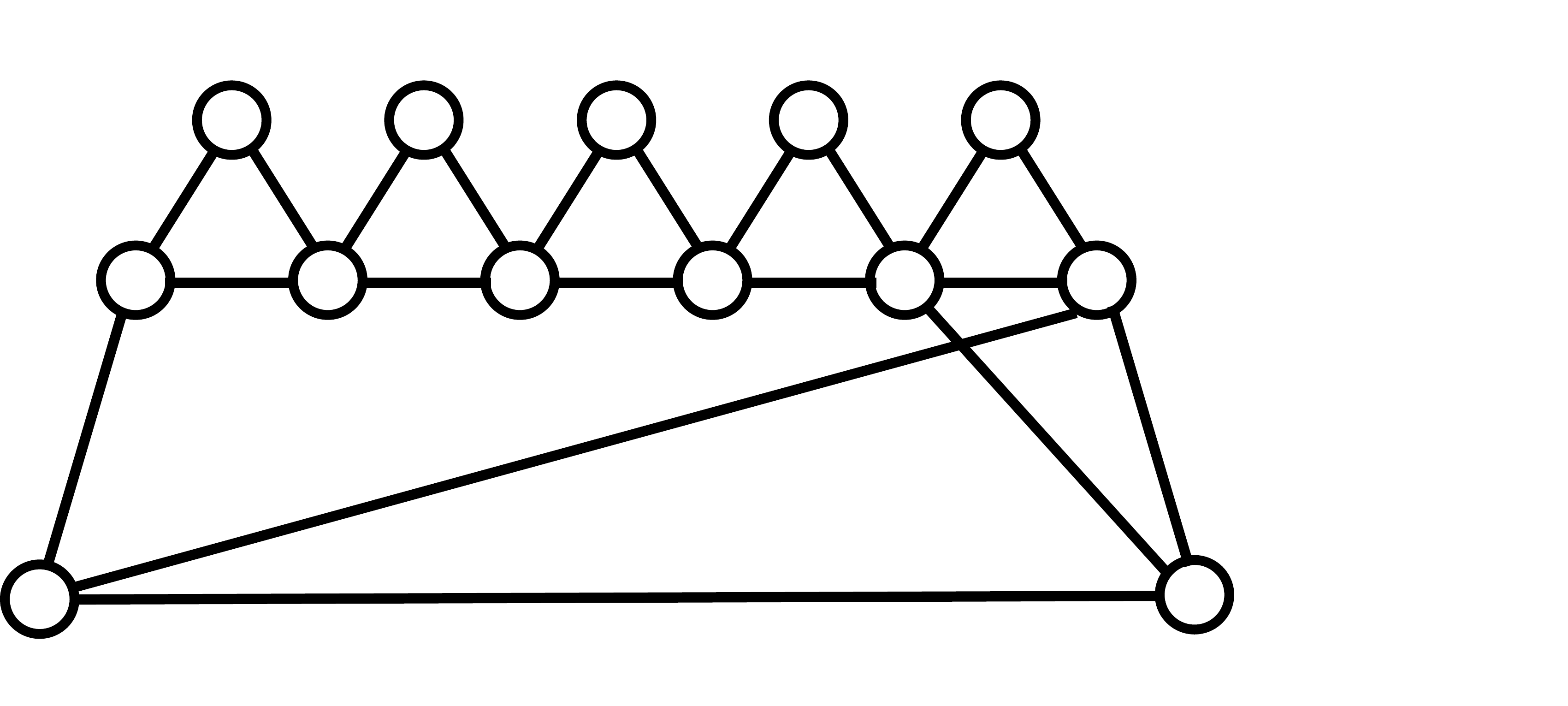}}%
    \put(-0.00058599,0.007203){\color[rgb]{0,0,0}\makebox(0,0)[lb]{\smash{$\bar v$}}}%
    \put(0.73921809,0.01187054){\color[rgb]{0,0,0}\makebox(0,0)[lb]{\smash{$\bar w$}}}%
    \put(0.07922884,0.21117422){\color[rgb]{0,0,0}\makebox(0,0)[lb]{\smash{$u_1$}}}%
    \put(0.18911061,0.21482815){\color[rgb]{0,0,0}\makebox(0,0)[lb]{\smash{$u_2$}}}%
    \put(0.30288175,0.21467202){\color[rgb]{0,0,0}\makebox(0,0)[lb]{\smash{$u_3$}}}%
    \put(0.5136209,0.23030826){\color[rgb]{0,0,0}\makebox(0,0)[lb]{\smash{$u_{k}$}}}%
    \put(0.72295978,0.27301621){\color[rgb]{0,0,0}\makebox(0,0)[lb]{\smash{$u_{k+1}$}}}%
    \put(0.10451158,0.41293451){\color[rgb]{0,0,0}\makebox(0,0)[lb]{\smash{$l_1$}}}%
    \put(0.22995841,0.41293451){\color[rgb]{0,0,0}\makebox(0,0)[lb]{\smash{$l_2$}}}%
    \put(0.35831557,0.41293451){\color[rgb]{0,0,0}\makebox(0,0)[lb]{\smash{$l_3$}}}%
    \put(0.60569483,0.41293451){\color[rgb]{0,0,0}\makebox(0,0)[lb]{\smash{$l_k$}}}%
    \put(0.49464958,0.42588077){\color[rgb]{0,0,0}\makebox(0,0)[lb]{\smash{...}}}%
    \put(0.43257139,0.22576027){\color[rgb]{0,0,0}\makebox(0,0)[lb]{\smash{...}}}%
  \end{picture}%
\endgroup%
\caption{Example instance of $G_k$ used to show our 1.5-approximation ratio for {\sc DPA} is tight.}
\label{fig:DPA-Tightness}
\end{figure}
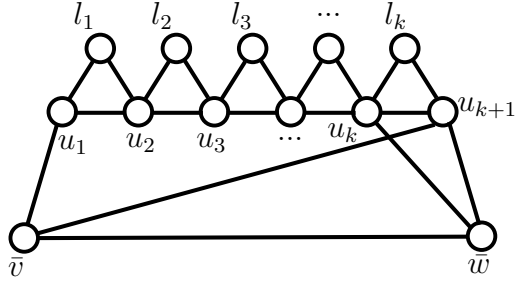
We prove this by giving a family $G_k$ of instances of {\sc SSC} with a bidirected input digraphs where our algorithm can choose arbitrarily close to $3|OPT(G_k)|/2$ stars. Our family of instances will only have stars of size one. Therefore, we can represent an instance of {\sc SSC} using only the corresponding digraph. Further, since the digraph must be bidirected, we can represent it using an undirected graph.

We define our family $G_k$ as follows: Let $V(G_k)=\{\bar v, \bar w\}\cup\{u_1,u_2,...u_{k+1}\}\cup\{l_1,l_2,...l_k\}$.
Let $E(G_k)$ contains $\bar v u_{k+1}$, $\bar w u_k$, and all edges in the cycle $\bar v,u_1,u_2,...u_{k+1}\bar w$ and the cycle $\bar v,u_1,l_1,u_2,l_2,...l_k,u_{k+1},\bar w$. An example instance of $G_k$ is depicted in Figure~\ref{fig:DPA-Tightness}

Suppose Algorithm~\ref{alg:1.5DPA} is run on $G_k$. 
When our cycle construction is run, it could build the path $u_{k+1},\bar w,u_k,...u_2,u_1,\bar v$ before terminating.
Then it would choose the perfect set corresponding to the cycle $\bar v, u_{k+1},\bar w,u_k, ... u_2,u_1$. This set has disjoint internal cuts $\{\bar v\}$ and $\{\bar w\}$. After this is contracted into a vertex $s$, the algorithm will have to choose the perfect set of size two contracting $l_i$ into $s$ for each $1\leq i\leq k$. Each of these sets have disjoint internal cuts $\{l_i\}$ and $V\setminus \{l_i\}$. Therefore our algorithm could choose $3k+3$ stars.

The optimal solution to $G_k$ will choose the perfect set corresponding to the Hamiltonian cycle $\bar v,u_1,l_1,u_2,l_2,...l_k,u_{k+1},\bar w$. This solution has objective $2k+3$. Then the approximation ratio achieved on $G_k$ could be as large as $\frac{3k+3}{2k+3}$. As $k$ approaches infinity, the ratio achieved on $G_k$ approaches $3/2$.


\section{Tightness of 1.6-Approximation Ratio for {\sc SSC}}
\label{sec:1.6SSC-Tightness}

\begin{figure}[t]
\centering
\def\svgwidth{0.80\textwidth}
\begingroup%
  \makeatletter%
  \providecommand\color[2][]{%
    \errmessage{(Inkscape) Color is used for the text in Inkscape, but the package 'color.sty' is not loaded}%
    \renewcommand\color[2][]{}%
  }%
  \providecommand\transparent[1]{%
    \errmessage{(Inkscape) Transparency is used (non-zero) for the text in Inkscape, but the package 'transparent.sty' is not loaded}%
    \renewcommand\transparent[1]{}%
  }%
  \providecommand\rotatebox[2]{#2}%
  \ifx\svgwidth\undefined%
    \setlength{\unitlength}{1036.78910778bp}%
    \ifx\svgscale\undefined%
      \relax%
    \else%
      \setlength{\unitlength}{\unitlength * \real{\svgscale}}%
    \fi%
  \else%
    \setlength{\unitlength}{\svgwidth}%
  \fi%
  \global\let\svgwidth\undefined%
  \global\let\svgscale\undefined%
  \makeatother%
  \begin{picture}(1,0.50915864)%
    \put(0,0){\includegraphics[width=\unitlength]{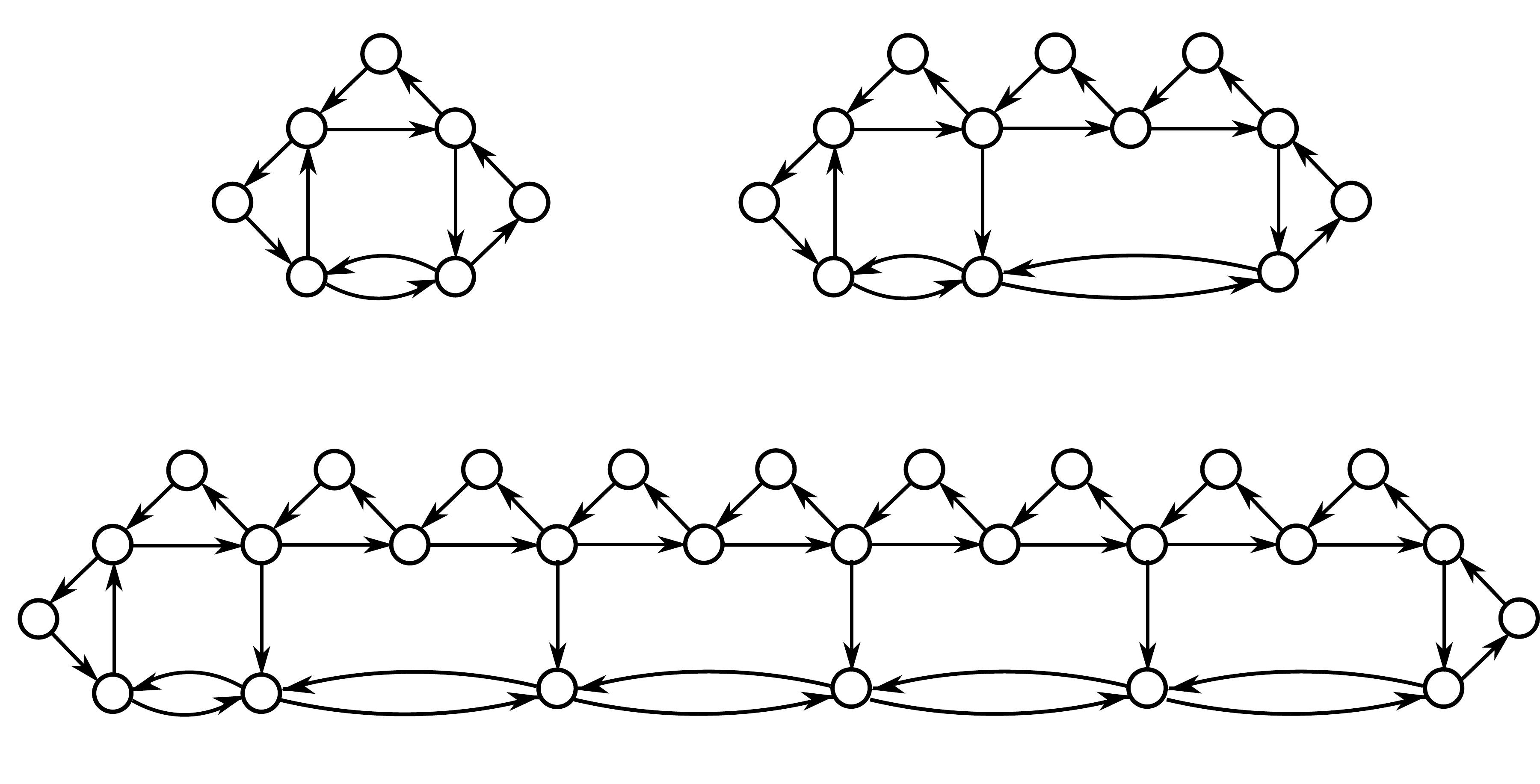}}%
    \put(0.2992524,0.44084508){\color[rgb]{0,0,0}\makebox(0,0)[lb]{\smash{$a$}}}%
    \put(0.22319341,0.48570879){\color[rgb]{0,0,0}\makebox(0,0)[lb]{\smash{$b$}}}%
    \put(0.1702828,0.43830979){\color[rgb]{0,0,0}\makebox(0,0)[lb]{\smash{$c$}}}%
    \put(0.12343486,0.38980837){\color[rgb]{0,0,0}\makebox(0,0)[lb]{\smash{$d$}}}%
    \put(0.16422013,0.30658432){\color[rgb]{0,0,0}\makebox(0,0)[lb]{\smash{$x$}}}%
    \put(0.30774015,0.30823784){\color[rgb]{0,0,0}\makebox(0,0)[lb]{\smash{$y$}}}%
    \put(0.35778477,0.38209224){\color[rgb]{0,0,0}\makebox(0,0)[lb]{\smash{$z$}}}%
    \put(0.62522333,0.44240399){\color[rgb]{0,0,0}\makebox(0,0)[lb]{\smash{$a'$}}}%
    \put(0.5653134,0.48570884){\color[rgb]{0,0,0}\makebox(0,0)[lb]{\smash{$b'$}}}%
    \put(0.51240278,0.43830984){\color[rgb]{0,0,0}\makebox(0,0)[lb]{\smash{$c'$}}}%
    \put(0.46555487,0.38980842){\color[rgb]{0,0,0}\makebox(0,0)[lb]{\smash{$d'$}}}%
    \put(0.50171045,0.30658432){\color[rgb]{0,0,0}\makebox(0,0)[lb]{\smash{$x'$}}}%
    \put(0.60625811,0.2942078){\color[rgb]{0,0,0}\makebox(0,0)[lb]{\smash{$y'$}}}%
    \put(0.15710649,0.17193191){\color[rgb]{0,0,0}\makebox(0,0)[lb]{\smash{$a'$}}}%
    \put(0.09719656,0.21523671){\color[rgb]{0,0,0}\makebox(0,0)[lb]{\smash{$b'$}}}%
    \put(0.04428592,0.16783771){\color[rgb]{0,0,0}\makebox(0,0)[lb]{\smash{$c'$}}}%
    \put(-0.002562,0.11933629){\color[rgb]{0,0,0}\makebox(0,0)[lb]{\smash{$d'$}}}%
    \put(0.03359359,0.03611219){\color[rgb]{0,0,0}\makebox(0,0)[lb]{\smash{$x'$}}}%
    \put(0.13814126,0.02373567){\color[rgb]{0,0,0}\makebox(0,0)[lb]{\smash{$y'$}}}%
    \put(0.22780519,0.27116425){\color[rgb]{0,0,0}\makebox(0,0)[lb]{\smash{(a)}}}%
    \put(0.67213688,0.27116425){\color[rgb]{0,0,0}\makebox(0,0)[lb]{\smash{(b)}}}%
    \put(0.47878708,0.00406905){\color[rgb]{0,0,0}\makebox(0,0)[lb]{\smash{(c)}}}%
  \end{picture}%
\endgroup%
\caption{Example instances of $T_k$ used to show our 1.6-approximation ratio for {\sc MSCS} and {\sc SSC} is tight. (a) $T_1$ (b) $T_2$ (c) $T_5$}
\label{fig:SSC-Tight}
\end{figure}
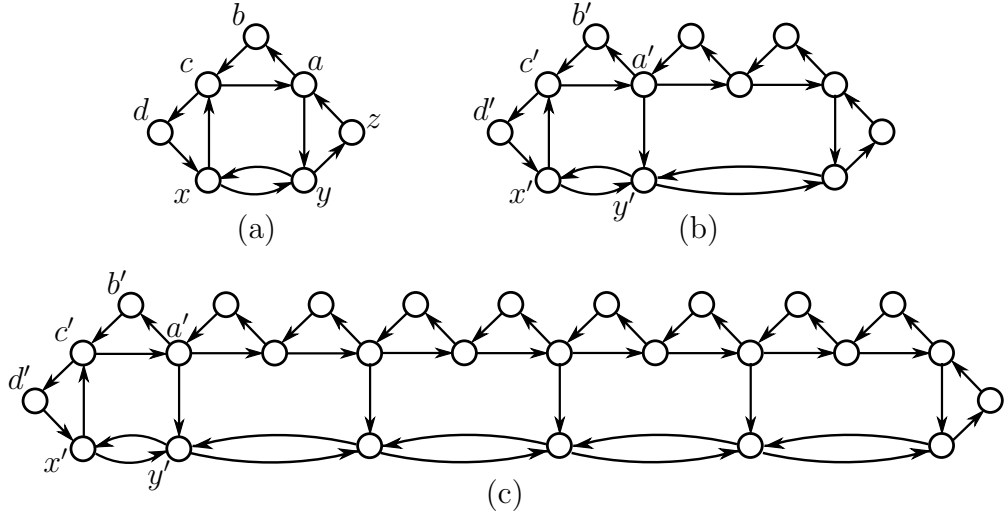

We show our ratio is tight by giving a simple family of digraphs where our algorithm may choose arbitrarily close to a 1.6-approximation. We use an example where all $|F|=1$ (i.e.\ when {\sc SSC} is equivalent to {\sc MSCS}). This allows us to describe any instance uniquely by giving its digraph. Figure~\ref{fig:SSC-Tight} gives examples of our family of digraphs, $T_k$. Formally, $T_k$ is recursively defined as follows: 

First, $T_1$ is a digraph with vertices $\{a, b, c, d, x, y, z\}$ and arcs of the cycles $abcdxyz$ and $yxca$. Each $T_k$ will have four specific vertices denoted by $c$, $d$, $x$ and $y$. To construct $T_{k+1}$ from $T_k$, we replace $x$ with the vertices $\{a',b',c',d',x',y'\}$. The arc from $x$ to $c$ is replaced with the arc from $a'$ to $c$. The arc from $d$ to $x$ is replaced with the arc from $d$ to $a'$. Similarly the two arcs between $x$ and $y$ are replaced with two arcs between $y'$ and $y$. Further, we connect these new vertices with the arcs of the paths $a'b'c'd'x'y'$ and $y'x'c'a'$. The vertices $c'$, $d'$, $x'$ and $y'$ from our expansion of $T_k$ are $c$, $d$, $x$ and $y$ for $T_{k+1}$ respectively.

We prove the following two lemmas about $T_k$ to establish the approximation ratio of $T_k$ approaches 1.6 as $k$ grows.

\begin{lemma} \label{lem:SSC-TightHam}
Every $T_k$ has a Hamiltonian cycle containing the path $dxy$ and of length $5k+2$.
\end{lemma}
\begin{proof}
We prove this by induction. By the definition of $T_1$, it contains the Hamiltonian cycle $abcdxyz$. Then for our inductive step, we assume there is such a Hamiltonian cycle $C$ in the digraph $T_k$.
We consider the path made by the arcs of $C$ in $T_{k+1}$ (note the arc $xy$ becomes the arc $y'y$ and $dx$ becomes $da'$). Then the arcs of $C$ form a path starting at $y'$, going through all vertices common with $T_k$ and ending at $a'$. Concatenating this with the path $a'b'c'd'x'y'$ will yield a Hamiltonian cycle in $T_{k+1}$. Note this cycle contains the path $d'x'y'$. We added five new arcs to this cycle, giving a total size of $5k+2+5=5(k+1)+2$, which completes our inductive proof.
 \end{proof}

\begin{lemma} \label{lem:SSC-TightA(I)}
Our algorithm may choose $8k+2$ arcs on input $T_k$.
\end{lemma}
\begin{proof}
We prove this by induction. For $T_1$, our cycle construction could build the path $cayx$. Then the cycle $cayx$ with internal cut $\{x\}$ may be used to create our perfect set. Let $w$ be the resulting supervertex after contracting these four vertices. The next three iterations of our algorithm will contract the cycles $wd$, $wb$ and $wz$. Total this choose 10 arcs, confirming our base case.

Now we assume our algorithm will produce a solution to $T_k$ using $8k+2$ arcs. Given an instance of $T_{k+1}$, consider the vertices added in our recursive construction: $\{a',b',c',d',x',y'\}$. As in our base case, the algorithm may contract the cycle $y'x'c'a'$ into a supervertex $w$. Then it can contract the cycles $wd'$ and $wb'$. After these contractions, the six vertices that replaced $x$ in $T_k$ have been combined to a single vertex. Then it follows that after our algorithm selects these 8 arcs and contracts, $T_{k+1}$ becomes an instance of $T_k$.
By our inductive assumption, this process could choose $8+(8k+2)=8(k+1)+2$ arcs.
 \end{proof}

From Lemma~\ref{lem:SSC-TightHam}, we know that the optimal solution to $T_k$ costs $5k+2$. Combining this result with Lemma~\ref{lem:SSC-TightA(I)}, we find $T_k$ could have an approximation ratio of $\frac{8k+2}{5k+2}$, which approaches $8/5=1.6$.

\end{document}